%% file: main.tex
\documentclass[11pt]{article}
\usepackage[top=1in,left=1in,right=1in,bottom=1in]{geometry}

\RequirePackage[tt=false, type1=true]{libertine}
\RequirePackage[varqu]{zi4}
\RequirePackage[T1]{fontenc}
\usepackage[numbers,sort]{natbib}

\usepackage{caption}

\usepackage{amsthm}  

\usepackage{style,color}

\newtheorem{lemma}{Lemma}[section]
\newtheorem{corollary}{Corollary}[section]
\newtheorem{definition}{Definition}[section]

\input{macro}

\begin{document}
\title{Parallel Write-Efficient Algorithms and Data Structures for Computational Geometry\footnote{This paper is the full version of a paper at SPAA 2018 with the same name.}}

\date{}

\author{
  Guy E. Blelloch\\CMU\\guyb@cs.cmu.edu \and
  Yan Gu\\CMU\\yan.gu@cs.cmu.edu \and
  Julian Shun\\MIT CSAIL\\jshun@mit.edu \and
  Yihan Sun\\CMU\\yihans@cs.cmu.edu
}



\maketitle

\input{abstract}

\input{intro}

\input{prelim}

\input{incremental}

\input{sort}

\input{delaunay}

\input{kdtree}

\input{augtree}

\section{Conclusions}

In this paper, we introduced new algorithms and data structures for computational geometry problems, including comparison sort, planar Delaunay triangulation, $k$-d trees, and static and dynamic augmented trees.
All of our algorithms, except for dynamic updates for augmented trees, are asymptotically optimal in terms of the number of arithmetic operations and writes to the large asymmetric memory, and have polylogarithmic depth. 

We introduced two frameworks for designing write-efficient parallel
algorithms.  The first one is for randomized incremental algorithms,
and combines DAG tracing and prefix doubling so that multiple objects
can be processed in parallel in a write-efficient manner.  The second
one is designed for augmented weight-balanced binary search trees,
where for dynamic insertions and deletions, we can reduce the
amortized number of writes compared to the standard data structure.
We believe that these techniques can be used for designing other
write-efficient algorithms.

\section*{Acknowledgments}
This work was supported in part by NSF grants CCF-1408940, CCF-1533858, and CCF-1629444.

\bibliographystyle{abbrv}

\input{bbl.bbl}

\appendix

\input{aug-appendix}

\end{document}

%% file: macro.tex
\newcommand{\work}{work}

\newcommand{\wcost}{\omega}

\newcommand{\smallmem}{small-memory}
\newcommand{\largemem}{large-memory}

\newcommand{\depth}{depth}

\newcommand{\id}[1]{\ifmmode\mathit{#1}\else\textit{#1}\fi}
\newcommand{\const}[1]{\ifmmode\mbox{\textc{#1}}\else\textsc{#1}\fi}
\newcommand{\anp}{Asymmetric NP}

\newcommand{\ourprob}{DAG tracing problem}
\newcommand{\ouralgo}{DAG tracing algorithm}
\newcommand{\ourprop}{tracable property}

\newcommand{\initround}{initial round}
\newcommand{\incround}{incremental round}

\newcommand{\incsort}{\textsc{incrementalSort}}

\newcommand{\kdtree}{$k$-d tree}
\newcommand{\batched}[1]{$#1$-batched incremental construction}
\newcommand{\batchedshort}[1]{$#1$-batched}

\newcommand{\func}[1]{{\sc #1}}
\newcommand{\am}{{{\mathbb{AM}}}}
\newcommand{\augt}{{{\mathbb{AT}}}}
\newcommand{\tree}{{{\mathbb{T}}}}

\newcommand{\bool}{{\text{\bf bool}}}

\newcommand{\tournament}{{priority tournament tree}}
\newcommand{\augval}{{augmented value}}

\newcommand{\labeling}{$\alpha$-labeling}
\newcommand{\critical}{critical}
\newcommand{\secondary}{secondary}

%% file: abstract.tex
\begin{abstract}

In this paper, we design parallel write-efficient geometric algorithms
that perform asymptotically fewer writes than standard algorithms for
the same problem.  This is motivated by emerging non-volatile memory
technologies with read performance being close to that of random
access memory but writes being significantly more expensive in terms
of energy and latency.  We design algorithms for planar Delaunay
triangulation, $k$-d trees, and static and dynamic augmented trees.
Our algorithms are designed in the recently introduced Asymmetric
Nested-Parallel Model, which captures the parallel setting in which
there is a small symmetric memory where reads and writes are unit cost
as well as a large asymmetric memory where writes are $\omega$ times
more expensive than reads.  In designing these algorithms, we
introduce several techniques for obtaining write-efficiency, including
DAG tracing, prefix doubling, reconstruction-based rebalancing and $\alpha$-labeling, which we believe
will be useful for designing other parallel write-efficient
algorithms.

\end{abstract}

%% file: intro.tex
\section{Introduction}

In this paper, we design a set of techniques and parallel algorithms
to reduce the number of writes to memory as compared to traditional
algorithms.  This is motivated by the recent trends in computer memory
technologies that promise byte-addressability, good read latencies,
significantly lower energy and higher density (bits per area) compared
to DRAM.  However, one characteristic of these memories is that
reading from memory is significantly cheaper than writing to it. Based
on projections in the literature, the asymmetry is between 5--40 in
terms of latency, bandwidth, or energy.  Roughly speaking, the reason
for this asymmetry is that writing to memory requires a change to the
state of the material, while reading only requires detecting the
current state.  This trend poses the interesting question of how to
design algorithms that are more efficient than traditional algorithms
in the presence of read-write asymmetry.


There has been recent research studying models and algorithms that
account for asymmetry in read and write costs~\cite{BT06, BBFGGMS16,
  BFGGS15, blelloch2016efficient, carson2016write, Chen11, Eppstein14,
  Gal05, jacob2017, ParkS09, Viglas12, Viglas14}.  Blelloch et
al.~\cite{BBFGGMS16, BFGGS15, blelloch2016efficient} propose models in
which writes to the asymmetric memory cost $\wcost{}\ge 1$ and all
other operations are unit cost.  The Asymmetric RAM
model~\cite{BBFGGMS16} has a small symmetric memory (a cache) that can
be used to hold temporary values and reduce the number of writes to
the large asymmetric memory.  The Asymmetric NP (Nested Parallel)
model~\cite{blelloch2016efficient} is the corresponding parallel
extension that allows an algorithm to be scheduled efficiently in
parallel, and is the model that we use in this paper to analyze
our algorithms.

Write-efficient parallel algorithms have been studied for many classes
of problems including graphs, linear algebra, and dynamic programming.
However, parallel write-efficient geometric algorithms have only been
developed for the 2D convex hull problem~\cite{BBFGGMS16}.
Achieving
parallelism (polylogarithmic depth) and optimal write-efficiency
simultaneously seems generally hard for many algorithms and data
structures in computational geometry.  Here, optimal write-efficiency
means that the number of writes that the algorithm or data structure
construction performs is asymptotically equal to the output size.  In
this paper, we propose two general frameworks and show how they can be
used to design algorithms and data structures from geometry with high
parallelism as well as optimal write-efficiency.

The first framework is designed for randomized incremental algorithms~\cite{CS89,Seidel93,Mulmuley94}.
Randomized incremental algorithms are relatively easy to implement in
practice, and the challenge is in simultaneously achieving high
parallelism and write-efficiency.  Our framework consists of two
components: a DAG-tracing algorithm and a prefix doubling technique.
We can obtain parallel write-efficient randomized incremental
algorithms by applying both techniques together.  The write-efficiency
is from the DAG-tracing algorithm, that given a current configuration
of a set of objects and a new object, finds the part of the
configuration that ``conflicts'' with the new object.  Finding $n$
objects in a configuration of size $n$ requires $O(n \log n)$ reads
but only $O(n)$ writes.  Once the conflicts have been found, previous
parallel incremental algorithms (e.g.~\cite{blelloch2016parallelism})
can be used to resolve the conflicts among objects taking linear reads
and writes.  This allows for a prefix doubling approach in which the
number of objects inserted in each round is doubled until all objects
are inserted.

Using this framework, we obtain parallel write-efficient algorithms for
comparison sort, planar Delaunay triangulation, and $k$-d trees, all
requiring optimal work, linear writes, and polylogarithmic depth.  The
most interesting result is for Delaunay triangulation (DT).  Although
DT can be solved in optimal time and linear writes sequentially using
the plane sweep method~\cite{blelloch2016efficient}, previous parallel
DT algorithms seem hard to make write-efficient.  Most are based on
divide-and-conquer, and seem to inherently require $\Theta(n \log n)$
writes.  Here we use recent results on parallel randomized incremental
DT~\cite{blelloch2016parallelism} and apply the above mentioned approach.
For comparison sort, our new algorithm is stand-alone (i.e., not based on other complicated algorithms like Cole's mergesort~\cite{Cole88,BFGGS15}).
For \kdtree{s}, we introduce the \batched{p} technique that maintains the balance of the tree while asymptotically reducing the number of writes.

The second framework is designed for augmented trees, including
interval trees, range trees, and priority search trees.  Our goal is
to achieve write-efficiency for both the initial construction as well as future
dynamic updates.  The framework consists of two
techniques.  The first technique is to decouple the tree construction
from sorting, and introduce parallel algorithms to construct the trees
in linear reads and writes after the objects are sorted (the sorting
can be done with linear writes~\cite{BFGGS15}).  Such algorithms
provide write-efficient constructions of these data structures, but
can also be applied in the rebalancing scheme for dynamic
updates---once a subtree is unbalanced, we reconstruct it.  The
second technique is \labeling{}.  We subselect some tree nodes as
\critical{} nodes, and maintain part of the augmentation only on these
nodes.  By doing so, we can limit the number of tree nodes that need to
be written on each update, at the cost of having to read more
nodes.\footnote{At a very high level, the \labeling{} is similar to
  the weight-balanced B-tree (WBB tree) proposed by Arge et
  al.~\cite{arge1999two,arge2003optimal}, but there are many
  differences and we discuss them in
  Section~\ref{sec:augtree}.}

Using this framework, we obtain efficient augmented trees in the asymmetric setting.
In particular, we can construct the trees in optimal work and writes, and polylogarithmic depth.
For dynamic updates, we provide a trade-off between
performing extra reads in queries and updates, while doing fewer writes on updates.
Standard algorithms use $O(\log n)$ reads and writes per update ($O(\log^2 n)$ reads on a 2D range tree).
We can reduce the number of writes by a factor of $\Theta(\log \alpha)$ for $\alpha\ge 2$, at a cost of increasing reads by at most a factor of $O(\alpha)$ in the worst case.
For example, when the number of queries and updates are about equal, we can improve the work by a factor of $\Theta(\log\wcost{})$, which is significant given that the update and query costs are only logarithmic.


The contributions of this paper are new parallel write-efficient
algorithms for comparison sorting, planar Delaunay triangulation,
$k$-d trees, and static and
dynamic augmented trees (including interval trees, range trees and
priority search trees).  We introduce two general
frameworks to design such algorithms, which we believe will be useful
for designing other parallel write-efficient algorithms.

\hide{
The first challenge is the complicatedness of many geometry algorithms, especially the parallel ones.
For instance, one of the simplest parallel algorithm for 2D Delaunay triangulation (DT) seems to be the incremental DT algorithm~\cite{blelloch2016parallelism}, although the analysis of that is already sophisticated.
An version that uses linear writes instead of $O(n\log n)$ writes is non-trivial due to the additional complexity by write-efficiency.
The second challenge is the varied combinations of problems and settings.
Taking the augmented trees for geometry problems as an example, there exist tens of various forms specifically designed for different queries and settings (e.g., online or offline).
Hence, it is possible to optimize the writes for each specific case, but that might not be the most efficient solution.
}

%% file: prelim.tex
\section{Preliminaries}
\label{sec:prelim}

\subsection{Computation Models}

\myparagraph{Nested-parallel model}
The algorithms in this paper is based on the nested-parallel model where a computation starts and ends with a
single \defn{root} task.  Each task has a constant number of
registers, and runs a standard instruction set from a random access
machine, except it has one additional instruction called FORK.
The FORK instruction takes an integer $n'$ and creates $n'$ \defn{child}
tasks, which can run in parallel.  Child tasks get a copy of the
parent's register values, with one special register getting an integer
from $1$ to $n'$ indicating which child it is.  The parent task
suspends until all its children finish at which point it continues
with the registers in the same state as when it suspended, except the
program counter advanced by one.  In this paper we consider the computation that has
\defn{binary branching} (i.e., $n'=2$).  In the model, a computation
can be viewed as a (series-parallel) DAG in the standard way.  We
assume every instruction has a weight (cost).  The \defn{work} ($W$)
is the sum of the weights of the instructions, and the \defn{\depth{}}
($D$) is the longest (unweighted) path in this DAG.

\myparagraph{\anp{} (Nested Parallel) model}
We use the {\anp{} (Nested Parallel) model}~\cite{BBFGGMS16}, which is the asymmetric version of the nested-parallel model, to measure the cost of an algorithm in this paper.
The memory in the \anp{} model consists of (i) an
infinitely large \emph{asymmetric} memory (referred to as \largemem{}) accessible to all processors and
(ii) a small private \emph{symmetric} memory (\smallmem{}) accessible only to
one processor.
The cost of writing to large memory is
$\wcost$, and all other operations have unit cost.
The size of the \smallmem{} is measured in words.  In this paper, we assume the
small memory can store a logarithmic number of words, unless specified
otherwise.
A more precise and detailed definition of the \anp{} model is given in~\cite{GuThesis}.

The \defn{work} $W$ of a computation is the sum of the costs of the
operations in the DAG, which is similar to the symmetric version but
just has extra charges for writes.  The \defn{\depth{}} $D$ is still the longest
unweighted path in the DAG.
Under mild assumptions, a work-stealing scheduler can execute an algorithm with
work $W$ and \depth{} $D$ in $W / p + O(pD)$
expected time on a round-synchronous PRAM with $p$ processors~\cite{BBFGGMS16}.  We assume
concurrent-read, and concurrent-writes use priority-writes to resolve
conflicts.  In our algorithm descriptions, the number of \defn{writes}
refers only to the writes to the \largemem{}, and does not include
writes to the \smallmem{}.  All reads and writes are to words of size
$\Theta(\log n)$-bits for an input size of $n$.


\subsection{Write-Efficient Geometric Algorithms}
Sorting and searching are widely used in geometry applications.
Sorting requires $O(\wcost{}n+n\log n)$ work and $O(\log^2 n)$ depth~\cite{BFGGS15}.
Red-black trees with appropriate rebalancing rules require $O(\wcost{}+\log n)$ amortized work per update (insertion or deletion)~\cite{Tarjan83a}.


These building blocks facilitate many classic geometric algorithms.
The planar convex-hull problem can be solved by first sorting the points by $x$ coordinates
and then using Graham's scan that requires $O(\wcost{}n)$ work~\cite{DCKO08}. 
This scan step can be parallelized with $O(\log n)$ \depth{}~\cite{goodrich1987finding}. 
The output-sensitive version uses $O(n\log h+\wcost{}n\log\log h)$ work and $O(\log^2 n\log h\log \log h)$ depth where $h$ is the number of points on the hull~\cite{BBFGGMS16}.


\hide{
\paragraph{One-dimension stabbing query and the interval tree.}

\paragraph{Two-dimension range query and the range tree.} Given a set of $n$ points $p=\{p_i =(x_i,y_i)\}$ where $p_i\in P=X\times Y$, the three-sided 2D range query ask for the list of (or weighted-sum, count, etc.) of points with x-coordinate in range $x_L$ and $x_R$, and y-coordinate in range $y_B$ and $y_T$.

This query can be answered with the \emph{range tree}, which is a two-level tree structure. (...) In particular it can be represented by a two-level augmented map as:

{\small
\begin{tabular}{l@{}@{ }l@{ }@{}l@{}@{ }l@{ }@{}l@{ }@{}l@{ }@{}l@{ }@{}l@{ }@{}l@{ }@{}l@{ }@{}l@{ }@{}l@{ }}
$R_I$ &$=$& $\am$&(&$P$, &$<_Y$, &$W$, &$W$, &$(k,v) \mapsto v$, &$+_W$, &$0_W$&)\\
$R_O$ &$=$& $\am$&(&$P$, &$<_X$, &$W$, &$R_I$, &$R_I.$\text{singleton}, &$\cup$, &$\emptyset$&)
\end{tabular}}

\julian{What is $W$? Why is it in the outer tree as well as the inner tree?}
When both level are implemented by augmented trees, this two-level map structure becomes a range tree.

\paragraph{Three-sided range query and the priority tree.}
Given a set of $n$ points $p=\{p_i =(x_i,y_i)\}$ where $p_i\in P=X\times Y$, the three-sided 2D range query ask for the list of (or weighted-sum, count, etc.) of points with x-coordinate in range $x_L$ and $x_R$, and y-coordinate no less than $y_B$.
We usually call the y-coordinate the \emph{priority} of that point.

This query can be answered with the priority search tree~\cite{McCreight85}, which costs $O(n\log n)$ work (and writes) \julian{then work is $O(\omega n\log n)$} in constructing a tree of size $n$, and $O(k+\log n)$ work per query (with output size $k$) \julian{work is $O(\omega k+\log n)$}.
The root of the priority tree stores the point in $P$ with the highest priority. Then all the other points in $P$ are evenly split into two parts by the median of their x-coordinate, which then recursively form the left and right subtree.
In each level of construction we scan all points to find the point with the highest y-coordinate, remove it, write the left and right part into new arrays, and recursive build the subtrees. In total it costs $O(n)$ work and writes per level. This guarantees the tree to be a complete binary tree.
With a straight-forward divide-and-conquer technique, the construction can be parallelized with depth $O(\log^2 n)$. This data structure is completely static.
}

%% file: incremental.tex
\section{General Techniques for Incremental Algorithms}
\label{sec:inc}

In this section, we first introduce our framework for randomized
incremental algorithms.  Our goal is to have a systematic approach for
designing geometric algorithms that are highly parallel and
write-efficient.


Our observation is that it should be possible to make randomized incremental algorithms write-efficient since each newly added object on expectation only conflicts with a small region of the current configuration.
For instance, in planar Delaunay triangulation, when a randomly chosen point is inserted, the expected number of encroached triangles is $6$.
Therefore, resolving such conflicts only makes minor modifications to the configuration during the randomized incremental constructions, leading to algorithms using fewer writes.
The challenges are in finding the conflicted region of each newly added object write-efficiently and work-efficiently, and in adding multiple objects into the configuration in parallel without affecting write-efficiency.
We will discuss the general techniques to tackle these challenges based on the \emph{history} graph~\cite{GKS92,BT93}, and then discuss how to apply them to develop parallel write-efficient algorithms for comparison sorting in Section~\ref{sec:inc-sorting}, planar Delaunay triangulation in Section~\ref{sec:delaunay}, and $k$-d tree construction in Section~\ref{sec:kdtree}.

\subsection{DAG Tracing}

We now discuss how to find the conflict set of each newly added object (i.e., only output the conflict primitives) based on a history (directed acyclic) graph~\cite{GKS92,BT93} in a parallel and write-efficient fashion.
Since the history graphs for different randomized incremental algorithms can vary, we abstract the process as a DAG tracing problem that finds the conflict primitives in each step by following the history graph.

\begin{definition}[\ourprob{}]\label{def:probdef}
  The \ourprob{} takes an element $x$, a DAG $G=(V,E)$, a root vertex $r\in V$
  with zero in-degree, and a boolean predicate function $f(x,v)$.
It computes the vertex set $S(G,x) = \{v \in V~|~f(x,v) \mbox{ and } \mbox{out-degree}(v) = 0\}$.
\end{definition}

We call a vertex $v$ \defn{visible} if $f(x,v)$ is true.

\begin{definition}[\ourprop{}]\label{def:dagproperty}
  We say that the \ourprob{} has the \ourprop{} when $v\in V$ is visible only if there exists at least one direct predecessor vertex $u$ of $v$ that is visible.
\end{definition}
\begin{center}
\begin{tabular}{cc}
  \toprule
  Variable & Description \\
  $D(G)$ & the length of the longest path in $G$ \\
  $R(G,x)$ & the set of all visible vertices in $G$ \\
  $S(G,x)$ & the output set of vertices \\
  \bottomrule
\end{tabular}
\end{center}
\begin{theorem}\label{thm:tracing}
The \ourprob{} can be solved in $O(|R(G,x)|)$ work, $O(D(G))$ depth and $O(|S(G,x)|)$ writes when the problem has the \ourprop{}, each vertex $v\in V$ has a constant degree, $f(x,v)$ can be evaluated in constant time, and the \smallmem{} has size $O(D(G))$.
Here $R(G,x)$, $D(G)$, and $S(G,x)$ are defined in the previous table.
\end{theorem}

\begin{proof}
We first discuss a sequential algorithm using $O(|R(G,x)|)$ work and $O(|S(G,x)|)$ writes.
Because of the \ourprop{}, we can use an arbitrary search algorithm to visit the visible nodes, which requires $O(R(G,x))$ writes since we need to mark whether a vertex is visited or not.
However, this approach is not write-efficient when $|S|=o(|R(G,x)|)$, and we now propose a better solution.

Assume that we give a global ordering $\prec_v$ of the vertices in $G$ (e.g., using the vertex labels) and use the following rule to traverse the visible nodes based on this ordering: a visible node $v\in V$ is visited during the search of its direct visible predecessor $u$ that has the highest priority among all visible direct predecessors of $v$.
Based on this rule, we do not need to store all visited vertices.
Instead, when we visit a vertex $v$ via a directed edge $(u,v)$ from $u$, we can check if $u$ has the highest priority among all visible predecessors of $v$.
This checking has constant cost since $v$ has a constant degree and we assume the visibility of a vertex can be verified in constant time.
As long as we have a \smallmem{} of size $O(D(G))$ that keeps the recursion stack and each vertex in $V$ has a constant in-degree, we can find the output set $S(G,x)$ using $O(|R(G,x)|)$ work and $O(|S(G,x)|)$ writes.

We note that the search tree generated under this rule is unique and deterministic.
Therefore, this observation allows us to traverse the tree in parallel and in a fork-join manner: we can simultaneously fork off an independent task for each outgoing edges of the current vertex, and all these tasks can be run independently and in parallel.
The parallel depth, in this case, is upper bounded by $O(D(G))$, the depth of the longest path in the graph.
\end{proof}

Here we assume the graph is explicitly stored and accessible, so we slightly modify the algorithms to generate the history graph, which is straightforward in all cases in this paper.

\subsection{The Prefix-Doubling Approach}

The sequential version of randomized incremental algorithms process one object (e.g., a point or vertex) in one iteration.
The prefix-doubling approach splits an algorithm into multiple rounds, with the first round
processing one iteration and each subsequent round doubling the number
of iterations processed.
This high-level idea is widely used in parallel algorithm design.
We show that the
prefix-doubling approach combined with the \ouralgo{} can reduce the
number of writes
by a factor of $\Theta(\log n)$ in a number of algorithms.  In particular,
our variant of prefix doubling first processes $n / \log n$ iterations
using a standard write-inefficient approach (called as the
\defn{\initround}).  Then the algorithm runs $O(\log\log n)$
\defn{\incround{s}}, where the $i$'th round processes the next
$2^{i-1}n/\log n$ iterations.

\hide{
\myparagraph{The \batched{p}}
The \batched{p} is a variant of the classic incremental construction when the iterative independence graph is a tree.
In the classic version, each primitive is directly inserted into the current configuration in turn.
However, in the \batchedshort{p} version, each primitive finds the leaf node it belongs to, but it leaves itself in the node, instead of directly adding to the configuration.
We start to process a leaf node once it contains $p$ primitives (here we assume that the \smallmem{} can hold all $p$ primitives).
After all primitives are added, we finish the computation in each node and generate the final output.
}

%% file: sort.tex
\section{Comparison Sort}\label{sec:inc-sorting}

We discuss a write-efficient version of incremental sort
introduced in~\cite{blelloch2016parallelism},
which can also be used in many geometry problems and algorithms.

The first algorithm that we consider is sorting by incrementally
inserting into a binary search tree (BST) with no rebalancing.
Algorithm~\ref{alg:incsort} gives pseudocode that works either
sequentially or in parallel.  In the parallel version, the
\textbf{for} loop is a \textbf{parallel for}, such that each vertex
tries to add itself to the tree in every round.  When there are
multiple assignments on Line~\ref{line:assign} to the same location,
the smallest value gets written using a priority-write.

\newcommand{\deref}[1]{^*\!#1}
\begin{algorithm}[t]
\caption{\incsort~\cite{blelloch2016parallelism}}\label{alg:incsort}
\fontsize{9pt}{9pt}\selectfont
\KwIn{A sequence $K = \{k_1,\ldots,k_n\}$ of keys.}
\KwOut{A binary search tree over the keys in $K$.}
  \vspace{.3em}
\tcp{\textrm{$\deref{P}$ reads indirectly through the pointer $P$.\\
The check on Line~\ref{line:par} is only needed for the parallel version.}}
  \vspace{.3em}
Root $\gets$ a pointer to a new empty location\\
\For {$i \leftarrow 1$ to $n$\label{line:forloop}} {
  $N\gets$ newNode($k_i$)\\
  $P\gets$ Root\\
  \While {true\label{line:whileloop}} {
    \If {$~\deref{P} =$ \emph{null} \label{line:isempty}} {
      write $N$ into the location pointed to by $P$ \label{line:assign}\\
      \If {$~\deref{P} = N$ \label{line:par}}{break}}
    \If {$N$.key $<{} \deref{P}$.key ~} {
        $P\gets$ pointer to $\deref{P}$.left}
    \Else {$P\gets$ pointer to $\deref{P}$.right}
  }
}
\Return {\emph{Root}}
\end{algorithm}

Blelloch et al.~\cite{blelloch2016parallelism} showed that the
parallel version of \incsort{} generates the same tree as the
sequential version, and for a random order of $n$ keys runs in $O(n
\log n)$ work and $O(\log n)$ depth with high probability\footnote{We
  say $O(f(n))$ \defn{with high probability (\whp{})} to indicate
  $O(kf(n))$ with probability at least $1- 1/n^k$.}  on a
priority-write CRCW PRAM.  
The depth bound increases to $O(\log^2 n)$ when only binary forking is allowed.
The key observation is that insertion of
$n$ keys into this tree in random order has 
the longest dependence chain to be
$O(\log n)$ \whp{} (i.e., the
tree has depth $O(\log n)$ \whp{}).  However, this algorithm yields
$O(n\log n)$ writes \whp{} as each element can execute the while loop
on Lines 5--13 $O(\log n)$ times \whp{}, with each iteration incurring
a write.
We discuss how the DAG-tracing algorithm and prefix doubling in Section~\ref{sec:inc} reduce the number of
writes in this algorithm.

\myparagraph{Linear-write and $O(\log^2 n\log \log n)$-depth incremental sort}
We discuss a linear-write parallel sorting algorithm based on the
prefix-doubling approach.
The \initround{} constructs the search tree
for the first $n / \log_2 n$ elements using
Algorithm~\ref{alg:incsort}. For the $i$'th \incround{} where $1\le
i\le\lceil\log_2\log_2 n\rceil$, we add the next $2^{i-1}n / \log_2 n$
elements into the search tree.  In an \incround{}, instead of directly running
Algorithm~\ref{alg:incsort}, we first find the correct position of
each element to be inserted (i.e., to reach line~\ref{line:assign} in
Algorithm~\ref{alg:incsort}).  This step can be implemented using the
\ouralgo{}, and in this case the DAG is just the search tree
constructed in the previous round.  The root vertex $r$ is the tree
root, and $f(x,v)$ returns true iff the the search of element $x$  visit the node $v$.
Note that the DAG is
actually a rooted tree, and also each element only visits one tree
node in each level and ends up in one leaf node (stored in $P$), which
means this step requires $O(2^{i-1}n)$ work, $O(\log n)$ depth \whp{},
and $O(2^{i-1}n/\log n)$ writes in the $i$'th round.

After each element finds the empty leaf node that it belongs to, in
the second step in this round we then run Algorithm~\ref{alg:incsort},
but using the pointer $P$ that was computed in the first step.  We
refer to the elements in the same empty leaf node as
belonging to the same \emph{bucket}.  Notice that the depth of this
step in one \incround{} is upper bounded by the depth of a random
binary search tree which is $O(\log n)$ \whp{}, so this algorithm has
$O(\log^2 n\log \log n)$ depth \whp{}: $O(\log \log n)$ rounds, and in each round there are $O(\log n)$ levels.

We now analyze the expected number of writes of in the second step.
In each \incround{} the number of elements inserted is the same as the
number of elements already in the tree.  Hence it is equivalent to
randomly throwing $k$ balls into $k$ bins, where $k$ is the number of
elements to be inserted in this \incround{}.  Assume that the
adversary picks the relative priorities of the elements within each
bin, so that it takes $O(b^2)$ work and to sort $b$ elements within each
bucket in the worst case.
We can show that the probability that there are
$b$ elements in a bucket is $\Pr(b)={k \choose b}\cdot
1/k^b(1-1/k)^{k-b}$, and $\Pr(b+1)<\Pr(b)\cdot c_1$ when $b>c_2$, for
some constant $c_1<1$ and $c_2>1$.  
The expected number of writes within each bucket in this \incround{} is therefore:
 $$k\cdot\sum_{i=0}^{k} i^2\Pr(i)<k\left(O(1)+\sum_{i=c_2}^{k}
i^2\cdot\Pr(c_2)c_1^{i-c_2}\right)=O(k)$$ Hence the overall number of
writes is also linear.
Algorithm~\ref{alg:incsort} sorts $b$ elements in a bucket with $O(b)$ depth, and \whp{}
the number of balls in each bin is $O(\log k)$, so the depth in this step is included in the depth analysis in the previous paragraph.
Combining the work and depth gives the following lemma.

\begin{lemma}
\textup{\incsort{}} for a random order of $n$ keys runs in
  $O(n \log n+\wcost{}n)$ expected work and $O(\log^2 n\log \log n)$ depth \whp{} on \anp{} model with priority-write.
\end{lemma}

\myparagraph{Improving the depth to $O(\log^2 n)$} We can improve the
depth to $O(\log^2 n)$ as follows.  Notice that for $b=c_3\log \log n$
and
$c_4>1$, $$\sum_{i=b}^{k}\Pr(i)<\sum_{i=b}^{k}\Pr(c_2)c_1^{i-c_2}=\log^{-c_4}k$$
This indicates that only a small fraction of the buckets in each
\incround{} are not finished after $b=c_3\log \log n$ iterations of
the while-loop on Line~\ref{line:whileloop} of
Algorithm~\ref{alg:incsort}.

In the depth-improved version of the algorithm, the while-loop
terminates after $b=c_3\log \log n$ iterations, and postpones these
insertions (and all further insertions into this subtree in future
rounds) to a final round.  The final round simply runs another round of Algorithm~\ref{alg:delaunaypar} and inserts all
uninserted elements (not write-efficiently).
Clearly the depth of the last round is $O(\log^2 n)$, since it is upper bounded
by the depth of running Algorithm~\ref{alg:incsort} for all $n$ elements.  The depth of the whole algorithm is
therefore $O(\log^2 n)+O(\log n\log\log n)\cdot O(\log\log n)+O(\log^2 n)=O(\log^2
n)$ \whp{}.

We now analyze the number of writes in the final round.  The
probability that a bucket in any round does not finish is
$\log^{-c_4}k$, and pessimistically there are in total $O(n\cdot \log^{-c_4}k)$ of such
buckets.  We also know that using Chernoff bound the maximum size of a
bucket after the first round is $O(\log^2 n)$ \whp{}.  The number of
writes in the last round is upper bounded by the overall number of
uninserted elements times the tree depth, which is $O(n)\cdot
\log^{-c_4}n\cdot O(\log^2 n)\cdot O(\log n)=o(n)$ by setting $c_3$
and $c_4$ appropriately large.  This leads to the main theorem.

\begin{theorem}
\textup{\incsort{}} for a random order of $n$ keys runs in
$O(n \log n+\wcost{}n)$ expected work and $O(\log^2 n)$ depth \whp{} on \anp{} model with priority-write.
\end{theorem}

Note that this gives a much simpler work/write-optimal
logarithmic-depth algorithm for comparison sorting than the
write-optimal parallel sorting algorithm in~\cite{BFGGS15} that is
based on Cole's mergesort~\cite{Cole88}, although our algorithm is
randomized and requires priority-writes.

%% file: delaunay.tex
\newcommand{\reptriangle}{\textsc{ReplaceTriangle}}
\newcommand{\seqinc}{\textsc{IncrementalDT}}
\newcommand{\parinc}{\textsc{ParIncrementalDT}}
\newcommand{\incircle}{\textsc{inCircle}}
\newcommand{\gt}{G_T}
\newcommand{\tri}{t}
\newcommand{\points}{V}
\newcommand{\point}{v}

\section{Planar Delaunay Triangulation}\label{sec:delaunay}

A Delaunay triangulation (DT) in the plane is a triangulation of a set
of points $P$ such that no point in $P$ is inside the circumcircle of
any triangle (the circle defined by the triangle's three corner
points).  We say a point \defn{encroaches} on a triangle if it is in
the triangle's circumcircle, so the triangle will be replaced once this point is added to the triangulation.
We assume for simplicity that the
points are in general position (no three points on a line or four
points on a circle).


Delaunay triangulation is widely studied due to its importance in many geometry applications.
Without considering the asymmetry between reads and writes, it can be
solved sequentially in optimal $\Theta(n \log n)$ work.
It is relatively easy to generate a sequential write-efficient version
that does $\Theta(n \log n)$ reads and only requires $\Theta(n)$ writes based on
the plane sweep method~\cite{blelloch2016efficient}.
There are several work-efficient parallel algorithms that run in polylogarithmic
depth~\cite{RS92,AG93,BHMT99,blelloch2016parallelism}. 
More practical ones (e.g.,~\cite{GKS92,BT93}) have linear depth.
Unfortunately, none of them perform any less than $\Theta(n \log n)$
writes.  In particular the divide-and-conquer
algorithms~\cite{AG93,BHMT99} seem to inherently require $\Theta(n
\log n)$ writes since the divide or merge step requires generating an
output of size $\Theta(n)$, and is applied for $\Theta(\log n)$
levels.  The randomized incremental approach of Blelloch
et al.~(BGSS)~\cite{blelloch2016parallelism}, which improves the Boissonnat and Teillaud algorithm~\cite{BT93} to polylogarithmic depth, also requires $O(n
\log n)$ writes for reasons described below.

In this section, we show how to modify the BGSS algorithm to use
only a linear number of writes, while
maintaining the expected $\Theta(n \log n)$ bound on work, and
polylogarithmic depth.
Algorithm~\ref{alg:delaunaypar} shows the pseudocode for the BGSS algorithm.
In the algorithm, the vertices are labeled from $1$ to $n$ and when
taking a $\min$ over vertices (Lines~7--\ref{line:reptri}) it is with respect to these labels.
The algorithm proceeds in rounds the algorithm
adds some triangles (Line~\ref{line:delaunayAdd}) and removes others
(Line~\ref{line:delaunayDel}) in each round.

In the algorithm, there are dependences between triangles so that
some of them need to be processed before the other triangles can proceed.
For a sequence of
points $\points$, BGSS define the dependence graph $\gt(\points) =
(T,E)$ for the algorithm in the following way.  The vertices $T$
correspond to triangles created by the algorithm, and for each call
to \reptriangle$(M,\tri,\point_i)$, we place an arc from triangle
$\tri$ and its three neighbors ($t_1, t_2$, and $t_3$) to each of the
one, two, or three triangles created by \reptriangle.  Every triangle
$T$ with depth $d(T)$ in $\gt(\points)$ is created by the algorithm in
round $d(T)$.  BGSS show that for a randomly ordered set of input
points of size $n$, the depth of the dependence graph is $O(\log n)$ \whp{}\footnote{We say $O(f(n))$ with
high probability (\whp{}) to indicate $O(c\cdot f(n))$ with probability to be at least $1-1/n^c$ for any constant $c>0$.}, and hence
the algorithm runs in $O(\log n)$ rounds \whp{}.  Each round can be done in
$O(\log n)$ depth giving an overall depth of $O(\log^2 n)$ \whp{} on the nested-parallel model.

The algorithm, however, is not write-efficient.  In particular, every
point moves down the DAG through the rounds (on
line~\ref{line:del-incircle}), and therefore can be moved $O(\log n)$
times, each requiring a write.

\SetKwInput{KwMaintains}{Maintains}%
\begin{algorithm}[t]
\caption{\parinc{}}
\label{alg:delaunaypar}
\fontsize{9pt}{9pt}\selectfont
    \SetKwFor{ParForEach}{parallel foreach}{do}{endfch}
\KwIn{A sequence $\points = \{\point_1,\ldots,\point_n\}$ of points in the plane.}
\KwOut{The Delaunay triangulation of $\points$.}%
\KwMaintains{$E(\tri)$, the points that encroach on each triangle $\tri$.}
  \vspace{.5em}
$\tri_b \gets$ a sufficiently large bounding triangle\\
$E(\tri_b) \gets \points$\\
$M \gets \{t_b\}$ \\
\While {$E(\tri) \neq \varnothing$ for any $\tri \in M$} {
    \ParForEach {triangle $\tri \in M$} {
        Let $\tri_1,t_2,t_3$ be the three neighboring triangles \\
        \If {$\min(E(\tri)) \leq \min(E(\tri_1) \cup E(\tri_2) \cup E(\tri_3))$} {
            \reptriangle{}$(M,t,\min(E(\tri)))$\label{line:reptri}
        }
    }
}
\Return {$M$}

  \vspace{.5em}
    \SetKwProg{myfunc}{function}{}{}
    \myfunc{\sc\reptriangle{}($M$,$\tri$,$v$)} {
    \ForEach {edge $(u,w)\in t$ (three of them)} {
        \If {$(u,w)$ is a boundary of $v$'s encroached region} {
            $\tri_o \gets$ the other triangle sharing $(u,w)$\\
            $\tri' \gets (u,w,v)$\label{line:new-triangle}\\
            \mbox{$E(\tri') \gets \{v' \in E(\tri) \cup E(\tri_o)~|~\mbox{\incircle{}}(v',t')\}$}\label{line:del-incircle}\\
            $M \gets M \cup \{t'\}$\label{line:delaunayAdd}
        }
    }
    $M \gets M \setminus \{t\}$\label{line:delaunayDel}\\
  }
\end{algorithm}

\myparagraph{A Linear-Write Version}
We now discuss a write-efficient version of the BGSS algorithm.  We use
the DAG tracing and prefix-doubling techniques introduced in
Section~\ref{sec:inc}.  The algorithm first computes the DT of the
$n/\log_2 n$ earliest points in the randomized order, using the non-write-efficient version.
This step requires linear writes.
It then runs $O(\log\log n)$ \incround{}s and in each round adds a number of points
equal to the number of points already inserted.

To insert points, we need to construct a search structure in the DAG tracing problem.
We can
modify the BGSS algorithm to build such a structure.  In fact, the
structure is effectively a subset of the edges of the dependence graph
$\gt(\points)$.  In particular, in the algorithm the only \incircle{} test
is on Line~\ref{line:del-incircle}.  In this test, to determine if a
point encroaches $t'$, we need only check its two ancestors $t$ and
$t_o$ (we need not also check the two other triangles neighboring
$t$, as needed in $\gt(\points)$).  This leads to a DAG with depth at most as large as
$\gt(\points)$, and for which every vertex has in-degree 2.  The
out-degree is not necessarily constant.  However, by noting that there
can be at most a constant number of outgoing edges to each level of the DAG,
we can easily transform it to a DAG with constant out-degree by
creating a copy of a triangle at each level after it has out-neighbors.  This does not increase the depth, and the number of copies
is at most proportional to the number of initial triangles ($O(n)$ in
expectation) since the in-degrees are constant.  We refer to this as
the \emph{tracing structure}. An example of this structure is shown in Figure~\ref{fig:delaunay-tracing}.

The tracing structure can be used in the \ourprob{} (Definition~\ref{def:probdef}) using the predicate $f(v,t) =
\mbox{\incircle{}}(v,t)$.  This predicate has the traceable property
since a point can only be added to a triangle $t'$ (i.e., encroaches on
the triangle) if it encroached one of the two input edges from $t$ and
$t_o$.  We can therefore use the \ouralgo{} to find all of the triangles
encroached on by a given point $v$ starting at the initial root
triangle $\tri_b$.

\begin{figure}
\begin{center}
  \includegraphics[width=.8\columnwidth]{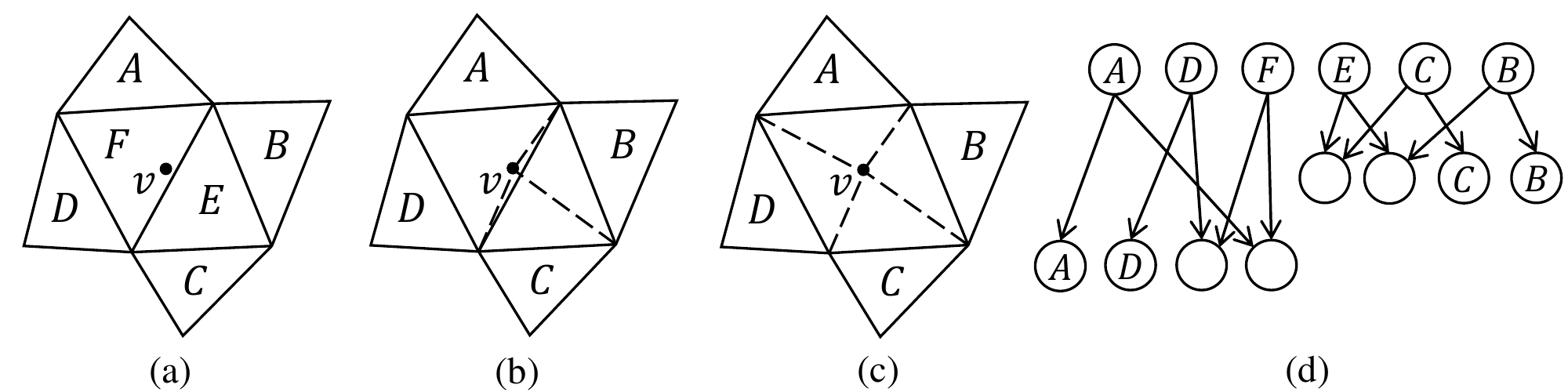}
\end{center}\vspace{-1em}
\caption[An example of the tracing structure in Delaunay
  triangulation]{An example of the tracing structure.  Here a point
  $v$ is added and the encroaching region contains triangles $E$ and
  $F$ (subfigure (a)).  Four new triangles will be generated and replace the two previous triangles.  They may or may not be created
  in the same round, and in this example this is done in two substeps
  (subfigures (b) and (c)).  Part of the tracing structure is shown in
  subfigure (d).  Four neighbor triangles $A$, $B$, $C$, and $D$ are
  copied, and four new triangles are created.  An arrow indicates that
  a point is encroached by the head triangle only if it is encroached
  by the tail triangle. }
\label{fig:delaunay-tracing}
\end{figure}

We first construct the DT of the first $n/\log_2 n$ points in the
\initround{} using Algorithm~\ref{alg:delaunaypar} while building
the tracing structure.  Then at the beginning of each \incround{},
each point traces down the structure to
find its encroached triangles, and leaves itself in the encroached set
of that triangle.   Note that the encroached set for a given point might
be large, but the average size across points is constant in expectation.
We now analyze the cost of finding all the encroached triangles when
adding a set of new points.  As discussed, the depth of $G$
is upper bounded by $O(\log n)$ \whp{}.  The number of encroached
triangles of a point $x$ can be analyzed by considering the degree of
the point (number of incident triangles) if added to the DT.  By
Euler's formula, the average degree of a node in a planar graph is at
most $6$.  Since we add the points in a random order, the expected value
of $|S(G,x)|$ in Theorem~\ref{thm:tracing} is constant.    Finally, the
number of all encroached (including non-leaf) triangles of this point
is upper bounded by the number of \incircle{} tests.  Then $|R(G,x)|$,
the expected number of visible vertices of $x$, is $O(\log n)$
(Theorem~4.2
in~\cite{blelloch2016parallelism}). 

After finding the encroached triangles for each point being added, we
need to collect them together to add them to the triangle.
This step can be done in parallel with a semisort,
which takes linear expected work (writes) and $O(\log^2 m)$ depth \whp{}~\cite{gu2015top}, where $m$ is the number of inserted points in this round. 
Combining these results leads to the following lemma.

\begin{lemma}\label{lem:delaunay-tracing}
  Given $2m$ points in the plane and a tracing structure $T$ generated
  by Algorithm~\ref{alg:delaunaypar} on a randomly selected subset of
  $m$ points, computing for each triangle in $T$ the points that
  encroach it among the remaining $m$ points takes
  $O(m\log m+\wcost{}m)$ \work{} ($O(m)$ writes) and $O(\log^2 n)$ depth \whp{}
  in the \anp{} model. 
\end{lemma}

The idea of the algorithm is to keep doubling the size of the set that we
add (i.e., prefix doubling).   Each round
applies Algorithm~\ref{alg:delaunaypar} to insert
the points and build a tracing structure, and then the \ouralgo{} to
locate the points for the next round.
The depth of each round is upper bounded by the overall depth of the
DAG on all points, which is $O(\log n)$ \whp{}, where $n$ is the original size.
We obtain the following theorem.

\begin{theorem}\label{the:delaunaypar}
Planar Delaunay triangulation can be computed using $O(n \log n+\wcost{}n)$ work (i.e., $O(n)$ writes) in expectation and $O(\log^2 n\log \log n)$ depth \whp{} on the \anp{} model with priority-writes.
\end{theorem}
\begin{proof}
  The original Algorithm~\ref{alg:delaunaypar}
  in~\cite{blelloch2016parallelism} has $O(\log^2 n)$ depth \whp{}.
  In the prefix-doubling approach, the depth of each round is no more
  than $O(\log^2 n)$, and the algorithm has $O(\log \log n)$ rounds.
  The overall depth is hence $O(\log^2 n\log \log n)$ depth \whp{}.

  The work bound consists of the costs from the \initround{}, and the
  \incround{s}.  The \initround{} computes the triangulation of the
  first $n/\log_2 n$ points, using at most $O(n)$ \incircle{} tests,
  $O(n)$ writes and $O(\wcost{}n)$ work.  For the \incround{s}, we
  have two components, one for locating encroached triangles in the
  tracing structure, and one for applying
  Algorithm~\ref{alg:delaunaypar} on those points to build the next
  tracing structure.    The first part is handled by
  Lemma~\ref{lem:delaunay-tracing}.
  For the second part we
we can apply a similar analysis to Theorem~4.2
of~\cite{blelloch2016parallelism}.
In particular, the probability that there is a dependence from
a triangle in the $i$'th point (in the random order) to a triangle
added by a later point at location $j$ in the ordering is upper bounded by $24/i$.
Summing across all points in the second half (we have already
resolved the first half) gives:
\[\mathbb{E}[C] \leq \sum_{i=m+1}^{2m} \sum_{j=i+1}^{2m} 24/i =
O(m)\; .\]
This is a bound on both the number of reads and the number of writes.
Since the points added in each round doubles, the cost is
dominated by the last round, which is $O(n \log n)$ reads and $O(n)$
writes, both in expectation.   Combined with the cost of the \initround{} gives the stated bounds.
\end{proof}

%% file: kdtree.tex
\section{Space-Partitioning Data Structures}\label{sec:kdtree}

Space partitioning divides a space into non-overlapping regions.\footnote{The other type of partitioning is object partitioning that subdivides the set of objects directly (e.g., R-tree~\cite{guttman1984r,manolopoulos2010r}, bounding volume hierarchies~\cite{gu2013efficient,wald2007ray}).}
This process is usually applied repeatedly until the number of objects in a region is small enough, so that we can afford to answer a query in linear work within the region.
We refer to the tree structure used to represent the partitioning as the space-partitioning tree.
Commonly-used space-partitioning trees include binary space partitioning trees, quad/oct-trees, \kdtree{s}, and their variants, and are widely used in computational geometry~\cite{DCKO08,har2011geometric}, computer graphics~\cite{akenine2008real}, integrated circuit design, learning theory, etc.

In this section, we propose write-efficient construction and update
algorithms for \kdtree{}s~\cite{bentley1975multidimensional}. We
discuss how to support dynamic updates write-efficiently in
Section~\ref{sec:kdupdate}, and we discuss how to apply our
technique to other space-partitioning trees in Section~\ref{sec:kdtree-ext}.



\subsection{$k$-d Tree Construction and Queries}\label{sec:kdconstruct}

\label{sec:kdconstruct}
\kdtree{s} have many variants that facilitate different queries.
We start with the most standard applications on range queries and nearest neighbor queries, and discussions for other queries are in Section~\ref{sec:kdtree-ext}.
A range query can be answered in $O(n^{(k-1)/k})$ worst-case work, and an approximate $(1+\epsilon)$-nearest neighbor (ANN) query requires $\log n\cdot O(1/\epsilon)^k$ work assuming bounded aspect ratio,\footnote{The largest aspect ratio of a tree node on any two dimensions is bounded by a constant, which is satisfied by the input instances in most real-world applications.}
both in $k$-dimensional space.
The tree to achieve these bounds can be constructed by always partitioning by the median of all of the objects in the current region either on the longest dimension of the region or cycling among the $k$ dimensions. The tree has linear size and $\log_2 n$ depth~\cite{DCKO08}, and can be constructed using $O(n\log n)$ reads and writes.
We now discuss how to reduce the number of writes to $O(n)$.



One solution is to apply the incremental construction by inserting the objects into a \kdtree{} one by one.
This approach requires linear writes, $O(n\log n)$ reads and polylogarithmic depth.
However, the splitting hyperplane is no longer based on the median, but the object with the highest priority pre-determined by a random permutation. 
The expected tree depth can be $c\log_2 n$ for $c>1$, but to preserve the range query cost we need the tree depth to be $\log_2n +O(1)$
(see details in Lemma \ref{lem:kdheight}).
Motivated by the incremental construction, we propose the following variant, called \batched{p}, which guarantees both write-efficiency and low tree depth.

\myparagraph{The \batched{p}}
The \batched{p} is a variant of the classic incremental construction where the dependence graph is a tree.
Unlike the classic version, where the splitting hyperplane (splitter) of a tree node is immediately set when inserting the object with the highest priority, in the \batchedshort{p} version, each leaf node will buffer at most $p$ objects before it determines the splitter.
We say that a leaf node \emph{overflows} if it holds more than $p$ objects in its buffer.
We say that a node is \emph{generated} when created by its parent, and \emph{settled} after finding the splitters, creating leaves and pushing the objects to the leaves' buffers.

The algorithm proceeds in rounds, where in each round it first finds the corresponding leaf nodes that the inserted objects belong to, and adds them into the buffers of the leaves. 
Then it settles all of the overflowed leaves, and starts a new round.
An illustration of this algorithm is shown in Figure~\ref{fig:kdtree}.
After all objects are inserted, the algorithm finishes building the subtree of the tree nodes with non-empty buffers recursively.
For write-efficiency, we require the \smallmem{} size to be $\Omega(p)$, and the reason will be shown in the cost analysis.

\input{kd-figure}



We make a partition once we have gathered $p$ objects in the corresponding subregion based on the median of these $p$ objects.
When $p=1$, the algorithm is the incremental algorithm mentioned above, but the range query cost cannot be preserved.
When $p=n$, the algorithm constructs the same tree as the classic $k$-d tree construction algorithm, but requires more than linear writes unless the \smallmem{} size is $O(n)$, which is impractical when $n$ is large.
We now try to find the smallest value of $p$ that preserves the query cost, and we analyze the cost bounds accordingly.

\mytitle{Range query}
We use the following lemma to analyze the cost of a standard $k$-d range query (on an axis-aligned hypercube for $k\ge 2$).

\begin{lemma}\label{lem:kdheight}
A $k$-d range query costs $O(2^{((k-1)/k) h})$ using our \kdtree{} of height $h$.
\end{lemma}
\begin{proof}[Proof Sketch]
A $k$-d range query has at most $2k$ faces that generate $2k$ half-spaces, and we analyze the query cost of each half-space.
Since each axis is partitioned once in every $k$ consecutive levels, one side of the partition hyperplane perpendicular to the query face will be either entirely in or out of the associated half-space.
We do not need to traverse that subtree (we can either directly report the answer or ignore it). Therefore every $k$ levels will expand the search tree by a factor of at most $2^{k-1}$.
Thus the query cost is $O(2^{((k-1)/k) h})$.
\end{proof}

\begin{lemma}\label{lem:kdp}
For our \batchedshort{p} \kdtree{}, $p=\Omega(\log^3 n)$ guarantees the tree height to be no more than $\log_2 n+O(1)$ whp.
\end{lemma}
\begin{proof}
We now consider the \batched{p}.
Since we are partitioning based on the median of $p$ random objects, the hyperplane can be different from the actual median.
To get the same cost bound, we want the actual number of objects on the two sides to differ by no more than a factor of $\epsilon$ \whp{}.
Since we pick $p$ random samples, by a Chernoff bound the probability that more than $1/2p$ samples are within the first $(1/2-\epsilon/4)n$ objects is upper bounded by $e^{-p\epsilon^2/24}$.
Hence, the probability that the two subtree weights of a tree node differ by more than a factor of $\epsilon$ is no more than $2e^{-p\epsilon^2/24}$.
This $\epsilon$ controls the tree depth, and based on the previous analysis we want to have $n(\frac{1}{2}+\frac{\epsilon}{4})^{\log_2 n/p+O(1)}< p$.
Namely, we want the tree to have no more than $\log_2 n/p+O(1)$ levels \whp{} to reach the subtrees with less than $p$ elements, so the overall tree depth is bounded by $\log_2 n/p+O(1)+\log_2 p=\log_2 n+O(1)$.
Combining these constraints leads to $\epsilon=O(1)/\log_2 n$ and $p=\Omega(\log^3 n)$.
\end{proof}

Lemma~\ref{lem:kdp} indicates that setting $p=\Omega(\log^3 n)$ gives
a tree height of $\log_2 n+O(1)$ \whp{}, and Lemma~\ref{lem:kdheight} shows
that the corresponding range query cost is $O(2^{((k-1)/k)
  (O(1)+\log_2 n)})=O(n^{(k-1)/k})$, matching the standard range query
cost.

\mytitle{ANN query}
If we assume that the input objects are well-distributed and the \kdtree{} satisfies the bounded aspect ratio,
then the cost of a $(1+\epsilon)$-ANN query is proportional to the tree height.
As a result, $p=\Omega(\log n)$ leads to a query cost of $\log n  \cdot O(1/\epsilon)^k$ \whp{}.\footnote{Actually the tree depth is $O(\log n)$ even when $p=1$.  However, for write-efficiency, we need $p=\Omega(\log n)$ to support efficient updates as discussed in Section~\ref{sec:kdupdate} that requires the two subtree sizes to be balanced at every node.}

\mytitle{Parallel construction and cost analysis}
To get parallelism, we use the prefix-doubling approach, starting with $n/\log n$ objects in the first round.
The number of reads of the algorithm is still $\Theta(n\log n)$, since it is lower bounded by the cost of sorting when $k=1$, and upper bounded by $O(n\log n)$ since the modified algorithm makes asymptotically no more comparisons than the classic implementation.
We first present the following lemma.
\begin{lemma}\label{lem:kdtree-overflow}
When a leaf overflows at the end of a round, the number of objects in its buffer is $O(p)$ \whp{} when $p=\Omega(\log n)$.
\end{lemma}
\begin{proof}[Proof Sketch]
In the previous round, assume $n'$ objects were in the tree. At that time no more than $p-1$ objects are buffered in this leaf node.
Then in the current round another $n'$ objects are inserted, and by a Chernoff bound, the probability that the number of objects falling into this leaf node is more than $(c+1)p$ is at most $e^{-{c^2p/2}}$. Plugging in $p=\Omega(\log n)$ proves the lemma.
\end{proof}

We now bound the parallel depth of this construction.
The initial round runs the standard construction algorithm on the first $n/\log_2 n$ objects, which requires $O((\log^2 p+\log n)\log n)=O(\log^2 n)$ depth.
Then in each of the next $O(\log \log n)$ \incround{s},
we need to locate leaf nodes and a parallel semisort to put the objects into their buffers.
Both steps can be done in $O(\log^2 n)$ depth \whp{}~\cite{gu2015top}.
Then we also need to account for the depth of settling the leaves after the incremental rounds.
When a leaf overflows, by Lemma~\ref{lem:kdtree-overflow} we need to split a set of $O(p)$ objects for each leaf, which has a depth of $O(\log^2 p)=O(\log \log n)$ using the classic approach, and is applied for no more than a constant number of times \whp{} by Lemma~\ref{lem:kdtree-overflow}.


We now analyze the number of writes this algorithm requires. The initial round requires $O(n)$ writes as it uses a standard construction algorithm on $n/\log_2 n$ objects. In the incremental rounds,
$O(1)$ writes \whp{} are required for each object to find the leaf node it belongs to and add itself to the buffer using semisorting~\cite{gu2015top}.
From Lemma~\ref{lem:kdtree-overflow}, when finding the splitting hyperplane and splitting the object for a tree node, the number of writes required is $O(p)$ \whp{}.
Note that after a new leaf node is generated from a split, it contains at least $p/2$ objects.
Therefore, after all \incround{s}, the tree contains at most $O(n/p)$ tree nodes, and the overall writes to generate them is $O((n/p) \cdot p)=O(n)$.
After the \incround{s} finish, we need $O(n)$ writes to settle the leaves with non-empty buffers, assuming $O(p)$ cache size. 
In total, the algorithm uses $O(n)$ writes \whp{}. 

\begin{theorem}\label{the:kd-construct}
A \kdtree{} that supports range and ANN queries efficiently can be computed using $O(n \log n+\wcost{}n)$ expected work (i.e., $O(n)$ writes) and $O(\log^2 n)$ depth \whp{} in the \anp{} model.  For range query the \smallmem{} size required is $\Omega(\log^3 n)$.
\end{theorem}

\subsection{$k$-d Tree Dynamic Updates}\label{sec:kdupdate}

Unlike many other tree structures, we cannot rotate the tree nodes in \kdtree{s} since each tree node represents a subspace instead of just a set of objects.
Deletion is simple for \kdtree{s}, since we can afford to reconstruct the whole structure from scratch when a constant fraction of the objects in the \kdtree{} have been removed, and before the reconstruction we just mark the deleted node (constant reads and writes per deletion via an unordered map).
In total, the amortized cost of each deletion is $O(\wcost{}+\log n)$.
For insertions, we discuss two techniques that optimize either the update cost or the query cost.

\myparagraph{Logarithmic reconstruction~\cite{overmars1983design}}
We maintain at most $\log_2 n$ \kdtree{s} of sizes that are increasing powers of $2$.
When an object is inserted, we create a \kdtree{} of size $1$ containing the object.
While there are trees of equal size, we flatten them and replace the two trees with a tree of twice the size.
This process keeps repeating until there are no trees with the same size.
When querying, we search in all (at most $\log_2 n$) trees.
Using this approach, the number of reads and writes on an insertion is $O(\log^2 n)$, and on a deletion is $O(\log n)$. The costs for range queries and ANN queries are $O(n^{(k-1)/k})$ and $\log^2 n\cdot O(1/\epsilon)^k$ respectively, plus the cost for writing the output. 


If we apply our write-efficient \batchedshort{p} version when reconstructing the \kdtree{s}, we can reduce the writes (but not reads) by a factor of $O(\log n)$ (i.e., $O(\log n)$ and $O(1)$ writes per update).

When using logarithmic reconstruction, querying up to $O(\log n)$ trees can be inefficient in some cases, so here we show an alternative solution that only maintains a single tree.

\myparagraph{Single-tree version}
As discussed in Section~\ref{sec:kdconstruct}, only the tree height affects the costs for range queries and ANN queries.
For range queries, Lemma~\ref{lem:kdp} indicates that the tree height should be $\log_2 n+O(1)$ to guarantee the optimal query cost.
To maintain this, we can tolerate an imbalance between the weights of two subtrees by a factor of $O(1/\log n)$, and reconstruct the subtree when the imbalance is beyond the constraint.
In the worst case, a subtree of size $n'$ is rebuilt once after $O(n'/\log n)$ insertions into the subtree.
Since the reconstructing a subtree of size $n'$ requires $O(n'\log n'+\wcost{}n')$ work, each inserted object contributes $O(\log n\log n'+\wcost{}\log n)$ work to every node on its tree path, and there are $O(\log n)$ such nodes.
Hence, the amortized work for an insertion is $O(\log^3 n+\wcost{}\log^2 n)$.
For efficient ANN queries, we only need the tree height to be $O(\log n)$, which can be guaranteed if the imbalance between two subtree sizes is at most a constant multiplicative factor.
Using a similar analysis, in this case the amortized work for an insertion is $O(\log^2 n+\wcost{}\log n)$.

\input{kdtree-ext}

%% file: kd-figure.tex
\begin{figure}
\begin{center}
  \includegraphics[width=.6\columnwidth]{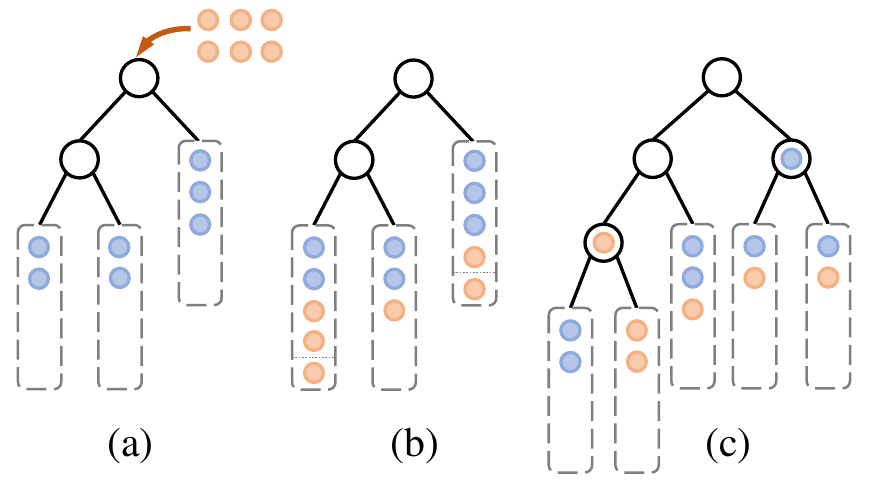}
\end{center}\vspace{-1.5em}
\caption[An illustration of the \batched{p}.]
{An illustration of one round in the \batched{p} for $p=4$.
Subfigure (a) shows the initial state of this round.
Then the new objects (shown in orange) are added to the buffers in the leaves, as shown in subfigure (b).
Two of the buffers overflow, and so we settle these two leaves as shown in subfigure (c).
}
\label{fig:kdtree}
\end{figure}

%% file: kdtree-ext.tex
\subsection{Extension to Other Trees and Queries}\label{sec:kdtree-ext}
In Section~\ref{sec:kdconstruct} we discussed the write-efficient
algorithm to construct a \kdtree{} that supports range and ANN
queries.  \kdtree{s} are also used in many other queries in real-world
applications, such as ray tracing, collision detection for
non-deformable objects, $n$-body simulation, and geometric culling
(using BSP trees).  The partition criteria in these applications are
based on some empirical heuristics (e.g., the surface-area
heuristic~\cite{goldsmith1987automatic}), which generally work well on
real-world instances, but usually with no theoretical guarantees.

The \batched{p} can be applied to these heuristics, as long as each
object contributes linearly to the heuristic.  Let us consider the
surface-area heuristic~\cite{goldsmith1987automatic} as an example,
which does an axis-aligned split and minimizes the sum of the two
products of the subtree's surface area and the number of objects.
Instead of sorting the coordinates of all objects in this subtree and
finding the optimal split point, we can do approximately by splitting
when at least $p$ of the objects are inserted into a region.  When
picking a reasonable value of $p$ (like $O(\log^2 n)$ or $O(\log^3
n)$), we believe the tree quality should be similar to the exact
approach (which is a heuristic after all).  However, the
\batchedshort{p} approach does not apply to heuristics that are not
linear in the size of the object set.  Such cases happen in the
Callahan-Kosaraju algorithm~\cite{callahan1995decomposition} when the
region of each \kdtree{} node shrinks to the minimum bounding box, or
the object-partitioning data structures (like R-trees or bounding
volume hierarchies) where each object can contribute arbitrarily to
the heuristic.

%% file: augtree.tex
\section{Augmented Trees}
\label{sec:augtree}
An \defn{augmented tree} is a tree that keeps extra data on each tree
node other than what is used to maintain the balance of this tree.
We refer to the extra data on each tree node as the \defn{augmentation}.
In this section, we introduce a framework that gives new algorithms
for constructing both static and dynamic augmented trees including
interval trees, 2D range trees,
and priority search trees that are parallel and write-efficient. Using these data
structures we can answer 1D stabbing queries, 2D range queries, and
3-sided queries (defined in
Section~\ref{sec:treeprelim}).
For all three problems, we assume that the query results need to be
written to the \largemem{}.  Our results are summarized in
Table~\ref{tab:treecost}.  We improve upon the traditional algorithms
in two ways. First, we show how to construct interval trees and
priority search trees using $O(n)$ instead of $O(n\log n)$ writes
(since the 2D range tree requires $O(n\log n)$ storage we cannot
asymptotically reduce the number of writes).  Second, we provide
a tradeoff between update costs and query costs
in the dynamic
versions of the data structures.  The cost
bounds are parameterized by $\alpha$.
By setting $\alpha=O(1)$ we achieve the same cost bounds as the traditional algorithms for queries
and updates.
$\alpha$ can be
chosen optimally if we know the update-to-query ratio $r$.  
For interval and priority trees, the optimal value of $\alpha$ is $\min(2+\wcost{}/r, \wcost{})$.
The overall work without considering writing the output can be improved by a factor of $\Theta(\log \alpha)$.
For 2D range trees, the optimal value of $\alpha$ is $2+\min(\wcost{}/r, \wcost{})/\log_2 n$.

\input{augtree-table}


We discuss two techniques in this section that we use to achieve
write-efficiency.  The first technique is to decouple the tree
construction from sorting, and we introduce efficient algorithms to
construct interval and priority search trees in linear reads and
writes after the input is sorted. Sorting can be done in parallel and
write-efficiently (linear writes).
Using
this approach, the tree structure that we obtain is perfectly
balanced.


The second technique that we introduce is the \labeling{} technique.  We mark a
subset of tree nodes as \emph{\critical{}} nodes by a predicate function
parameterized by $\alpha$, and only maintain augmentations
on these \critical{} nodes.  We can then guarantee that every update
only modifies $O(\log_\alpha n)$  nodes, instead of $O(\log n)$
nodes as in the classic algorithms.  At a high level, the
\labeling{} is similar to the weight-balanced B-tree (WBB tree)
proposed by Arge et al.~\cite{arge1999two,arge2003optimal} for the
external-memory (EM) model~\cite{AggarwalV88}. However,
as we discuss in Section~\ref{sec:aug-dyn}, directly applying the EM
algorithms~\cite{arge1999two,arge2003optimal,ajwani2010geometric,sitchinava2012a,Sitchinava2012b} does not give us the desired bounds in our model.
Secondly, our underlying tree is still binary.
Hence, we mostly need no changes to the algorithmic part that dynamically maintains the augmentation in this trees, but just relax the balancing criteria so the underlying search trees can be less balanced.
An extra benefit of our framework is that bulk
updates can be supported in a straightforward manner.    Such bulk updates seem complicated and less obvious in previous approaches. 
We propose algorithms on our trees that can support bulk updates
write-efficiently and in polylogarithmic depth.



The rest of this section is organized as follows. We first provide the
problem definitions and review previous results in
Section~\ref{sec:treeprelim}. Then in Section~\ref{sec:presort}, we
introduce our post-sorted construction technique for constructing
interval and priority search trees using a linear number of
writes. Finally, we introduce the \labeling{} technique to support a
tradeoff in query and update cost for interval trees, priority search
trees, and range trees in Section~\ref{sec:aug-dyn}.


\input{treeprelim}

\input{presort}

\input{dynamic}

\hide{

\subsection{The Priority Tournament Tree}
Here we propose a new augmented tree structure addressing the 3-sided range query, which not only cost linear writes in construction but also supports efficient dynamic updates. This data structure is a hybrid of priority tree and tournament tree, and thus we call it the \emph{\tournament{}}. It organizes all points in a search tree ordered by the x-coordinates. Then we assign an \augval{} to each tree node top-down, which is the highest priority in its subtree that does not appear as the augmented value in any of its ancestor(s). In particular, the augmented values in the tree forms exactly the same configuration as the pre-sorted priority tree as introduced in Section \ref{sec:presort}, which is also a heap on y-coordinates. The difference is that the \tournament{} also stores points in internal nodes, instead of just the replicates of the leaves as in the pre-sorted priority tree.

\paragraph{\textbf{Construction.}}
The pseudocode of constructing a \tournament{} is shown in Algorithm \ref{alg:tournamentbuild}. We first sort all points in $p$ by x-coordinate (Line \ref{line:tournamentsort}). This is also the order maintained by the search tree.
We use the middle point ($p_{n/2}$) as the root (Line \ref{line:tournamentroot}), such that the left and right subtree must be weight-balanced\footnote{Using similar method the tree can also be maintained valid as other balancing schemes such as AVL trees, red-black trees, using the \texttt{join} function\cite{blelloch2016just}.}.
Then we find the highest priority $m$ in $p$ and its corresponding index $k$ (Line \ref{line:tournamentmax}). This means that the \augval{} of the root is $m$. We then change the priority of $p_k$ as $-\infty$ (Line \ref{line:tournamentreset}) such that it will never appear as the \augval{} in any lower levels. We then recursively build the left subtree on all points on the left of $p_{n/2}$ and the right subtree on all points on the right (Line \ref{line:tournamentrecursive}). In this process, since we may need to reset the priority to $-\infty$ for some points to guarantee correctness, we copy the priority of each node in advance (Line \ref{line:tournamentcopy}).

In each invocation to \func{FromSorted}, the \func{ParMax} costs $O(n)$ work and $O(\log n)$ depth. All the other operations only have constant work.
There are in total $\log n$ levels of recursion. In all, the construction has $O(n)$ work, $O(n)$ space, $O(\log^2 n)$ depth and $O(n)$ writes.

In the construction process, we can also easily record the location (the tree node) from which each \augval{} comes, as well as the path to track it (e.g., we can store the x-coordinate of the corresponding point, and it can be easily located in the search tree). This is important to guarantee the efficiency of updates. For simplicity, we do not show this process in the code in Algorithm \ref{alg:tournamentbuild}.

\begin{algorithm}[t]
\caption{\func{BuildPTTree}$(A,n)$}
\label{alg:tournamentbuild}
    \SetKwFor{ParForEach}{parallel foreach}{do}{endfch}
\SetKwInOut{Note}{Note}
\Note{``$||$'' means that two recursive calls can run in parallel. $\langle m,k\rangle=$\func{ParMax}$(p,f)$ is a parallel function returning the maximum value in $p$ (according to comparison function $f$) as $m$, and the corresponding index as $k$. \func{Node}$(p)$ returns a new tree node of entry $p$.}

\KwIn{An array $p = \{(x_1,y_1),\dots,(x_n,y_n)\}$ of all points ($\in P$).}
\KwOut{The root of the \tournament{} containing all points in $A$.}%
$f(p_1,p_2):(P\times P\mapsto \bool) = (p_1.x<p_2.x)$;\\
$p=$\func{Sort}$(p,f)$\\
\label{line:tournamentsort}
$p'=$ an array of elements of type $P'=X\times Y\times Y$\\
\ParForEach{$(x_i,y_i)\in A$} {
\label{line:tournamentcopy}
$p'[i].x=x_i,p'[i].y=y_i,p'[i].y'=y_i$
}
\Return {\func{FromSorted}$(A)$}\\
  \vspace{.5em}
\SetKwProg{myfunc}{function}{}{}
\myfunc{\func{FromSorted}$(A,0,n)$} {
$f:(P'\times P'\mapsto \bool) = (p_1.y'<p_2.y')$;\\
$\langle m,k\rangle = $ParMax$(A, f)$;\\
\label{line:tournamentmax}
$r=$\func{Node}$(A[n/2])$;\\
\label{line:tournamentroot}
$r.$\func{aug}$=m$; $A[k].y'=-\infty$;\\
\label{line:tournamentreset}
$r.$\func{left}$=$\func{FromSorted}$(A,0,n/2-1)$ $||$ $r.$\func{right}$=$\func{FromSorted}$(A,n/2+1,n)$\\
\label{line:tournamentrecursive}
\Return {$r$}
}
\end{algorithm}

\paragraph{\textbf{Query.}} To answer the 3-sided query $(x_L,x_R,y_B)$, we first search $(x_L,x_R)$ as a range in the search tree, finding all related nodes and subtrees (at most $O(\log n)$ each). For each such nodes, we check the entry stored in it, as well as the point that contributes to its augmented value. Here we use ``check'' to mean determining if a point is in the 3-sided query window. We need to be careful that if the priority stored in the entry is higher than the augmented value, we \emph{do not} report it because it should have been reported on higher levels. For each related subtrees, as mentioned, the augmented values in the tree structure form a heap. Because all the nodes in the subtree must lie in the query range on x-coordinate, we can simply check all nodes with augmented values higher than $y_B$. We check both the entry in such nodes and the point that contributes to its augmented value. Whenever we visit a node with the augmented value less than $y_B$, we skip the whole subtree. 
Assume that the output size is $k$. The number of all access to the ``heap'' nodes is no more than $4k$ (all the output nodes, the ones that store their priorities as augmented values, and their two children). The number of all single nodes that need to be checked is at most $O(\log n)$. Thus the query can be answered in $O(k+\log n)$ work.

\begin{figure}
  \includegraphics[width=0.8\columnwidth]{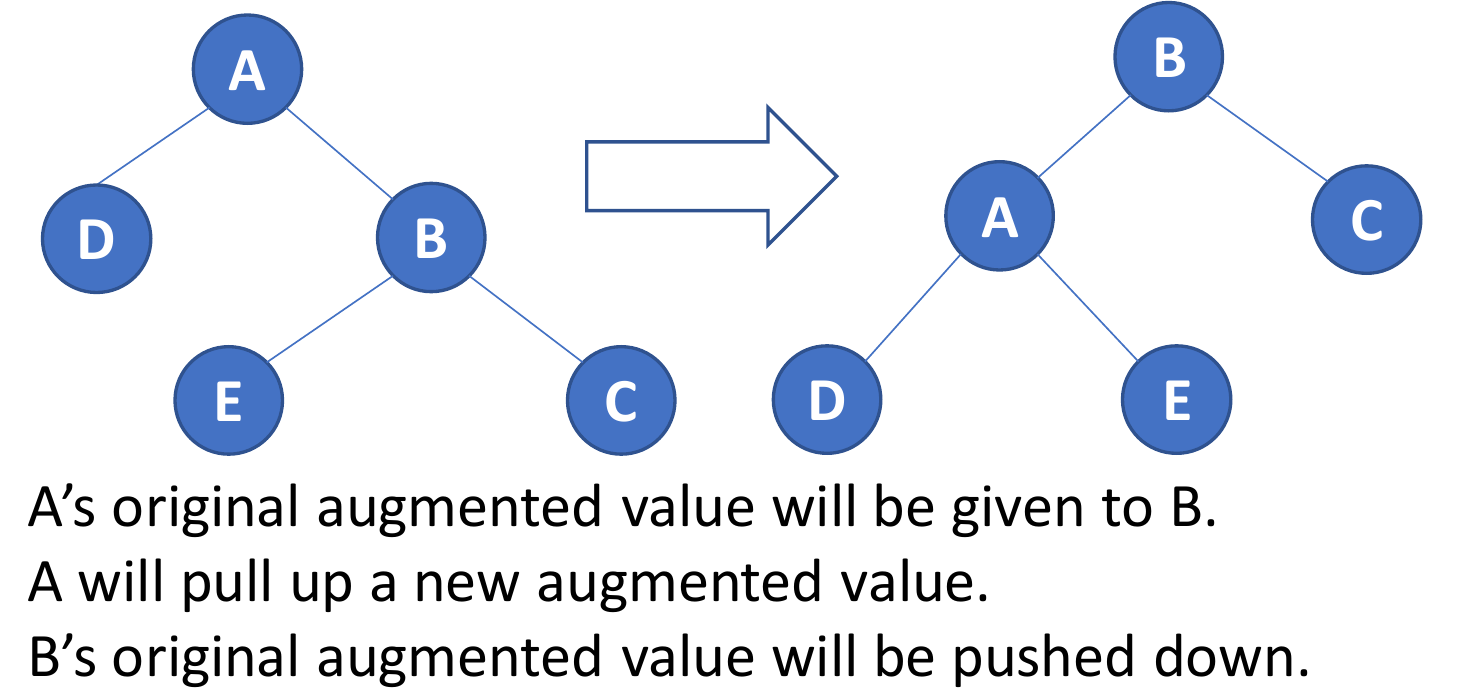}\\
  \caption{Single rotation illustration.}\label{fig:rotation}
\end{figure}

\paragraph{\textbf{Update.}} If no rotation is required, we only need update the augmented value on the insertion path, which costs $O(\log n)$ writes. When an imbalance occurs then, we use any standard insertion algorithm of the corresponding balancing scheme using rotations.
In this case, the augmented values of most of the nodes involved in the rotations need to be updated.
Here we give an algorithm for a single-rotation, and double rotations will be viewed as two consecutive single-rotations. We need two functions. One is to pull up an \augval{} of a certain node from its descendants (assigning a new \augval{} to some node). The other is to push the \augval{} of the current node down to one of its descendants (after a node's \augval{} is updated, its original \augval{} should be ``given back'' to some of its descendants). When a rotation occurs, e.g., as shown in Figure \ref{fig:rotation}, there are three steps. First, $B$'s new \augval{} should be directly from $A$'s old \augval{}. In this case $B$'s old \augval{} is now eligible to appear in some lower levels in the subtrees rooted at $E$ or $C$. Thus we push $B$'s \augval{} down. Then we set $B$'s new \augval{} to $A$'s old \augval{} because they represent the same subtree. Finally, now $A$ has new children $D$ and $E$, so it need to pull up a new augmented value from one of its new descendants.

To pull up an \augval{}, we compare the priority in the node itself and the \augval{}s in its two children. If the highest one $m$ is the one from a child, we set the current \augval{} to $m$, and recursively pull up a new \augval{} for this child.
To push down an \augval{} $a$, we go to the branch where the priority comes from, set the corresponding child's \augval{} to be $a$, and recursively push this child's old \augval{} down. Both functions cost $O(\log n)$, which is also the asymptotic cost of a single-rotation.

For a weight-balanced tree or red-black tree, the amortized number of rotations per update is constant (do we need a lemma?).  Thus the amortized work (and also the number of writes) of one update to the \tournament{} is $O(\log n)$. In total the cost is $O(\omega\log n)$

By using the $(\alpha,2\alpha)$-WBB tree, the read cost in query or any searching increases by a $O(\frac{\alpha}{\log \alpha})$ factor, as a payoff we can make the amortized number of write per update $\frac{\log n}{\log \alpha}$. By choosing $\alpha = O(\omega)$, we save a $O(\log \omega)$ factor per update at the cost of slightly increasing the query cost.

\subsection{The Dynamic Range Tree}
In this section we discuss how to make update on a range tree (as defined in Section \ref{sec:treeprelim}) write-efficient. We first analyze the cost on a standard range tree as defined in Section \ref{sec:treeprelim}. We use red-black tree as the inner tree. It has the property that each update only takes amortized constant number of rotations, and the amortized number of nodes that need to change the balancing information (color and height) is also a constant~\cite{??}.
This means that only a constant number of writes is required for each inner tree per update.
To this point we have made $O(\log^2 n)$ reads, and $O(\log n)$ writes.
Now the outer tree is also valid as a range tree, but may be unbalanced. To re-balance the tree we also employ corresponding rebalancing algorithm of the outer balancing scheme, but whenever a rotation occurs, we re-compute the augmented value (the inner tree) of the involved nodes, taking $O(n)$ writes (also reads) for a subtree of size $n$. Meanwhile, for a subtree with size $n$, imbalance occur (in the worst case) every $O(n)$ updates, so the amortized cost on each level is a constant. We accumulate the cost across all levels, and the total (amortized) number of rotations of a single update is $O(\log n)$. In all the amortized cost of a single update is $O(\log^2 n + \wcost \log n)$.

To reduce the number of writes per update, we use a $(\alpha,2\alpha)$-WBB tree as the outer tree. By doing this the searching time on the outer tree change from $O(\log n)$ to $O(\alpha \frac{\log n}{\log \alpha})$, because the tree height gets shallower (from $O(\log n)$ to $O(\log n/\log \alpha)$), but each visit of a tree node costs $O(\alpha)$ instead of a constant. Thus the query time becomes $\wcost k + \alpha \frac{\log^2 n}{\log \alpha}$. As for updates, the number of reads increase to $O(\omega\frac{\log^2 n}{\log \alpha})$. If not considering the rebalancing on the outer tree, the number of writes is $O(\log n/\log \alpha)$ (constant cost on each inner tree, and $O(\log n/\log \alpha)$ inner trees involved). Each $(\alpha,2\alpha)$-WBB subtree gets unbalanced no less than every $O(n)$ updates. Recomputing the augmented value on such a subtree takes $O(n)$ time. Using similar analysis as the binary weight-balanced tree, the amortized cost per update is proportional to the tree height, which is $O(\log n/\log \alpha)$. In all the work of an update is $O\left(\frac{\log n}{\log \alpha}\left( \alpha \log n+\wcost \right)\right)$. When $\wcost > \log n$, we choose $\alpha = \wcost/\log n$, otherwise we choose $\alpha=2$, which make the tree just a binary weight balanced tree.

\hide{\paragraph{Persistent tree using path-copying.} When implemented with path-copying, the tree structure can also be fully \emph{persistent} (or \emph{functional}), which means that any updates will not modify existing tree but create a new version, with no extra work asymptotically. }

}

%% file: augtree-table.tex
\begin{table*}
\small
\def\arraystretch{1.1}
\begin{tabular}{rccc}
\toprule
&\textbf{Construction} &\textbf{Query}&\textbf{Update}\\
\midrule
\textbf{Classic interval tree}       &$O(\wcost n\log n)$   &$O(\wcost k+\log n)$& $O(\wcost{}\log n)$\\
\textbf{Our interval tree}           &$O(\wcost n+n\log n)$ &$O(\wcost k+\alpha\log_\alpha n)$&$O((\wcost +\alpha)\log_\alpha n)$ \\
\midrule
\textbf{Classic priority search tree}  &$O(\wcost n\log n)$  &$O(\wcost k+\log n)$& $O(\wcost{}\log n)$ \\
\textbf{Our priority search tree}      &$O(\wcost n+n\log n)$&$O(\wcost k+\alpha\log_\alpha n)$&$O((\wcost +\alpha)\log_\alpha n)$\\
\midrule
\textbf{Classic range Tree}   &$O(\wcost n\log n)$&$O(\wcost k + \log^2 n)$&$O((\log n+\wcost)\log n)$\\
\textbf{Our range tree}       &$O((\alpha+\wcost)\, n\log_\alpha n)$&$O(\wcost k + \alpha \log_{\alpha}n\,{\log n})$&$O(( \alpha \log n+\wcost )\log_{\alpha}n)$\\
\bottomrule
\end{tabular}
\caption{A summary of the work cost of the data structures discussed in Section~\ref{sec:augtree}.
In all cases, we assume that the tree contains $n$ objects (intervals or points).
For interval trees and priority search trees, we can reduce the number of writes in the construction from $O(\log n)$ per element to $O(1)$.
For dynamic updates, we can reduce the number of writes per update by a factor of $\Theta(\log\alpha)$ at the cost of increasing the number of reads in update and queries by a factor of $\alpha$ for any $\alpha\ge 2$.
}\label{tab:treecost}
\end{table*}

%% file: treeprelim.tex
\subsection{Preliminaries and Previous Work}
\label{sec:treeprelim}
We define the \emph{weight} or \emph{size} of tree node or a subtree as the number of nodes in this subtree plus one.
The ``plus one'' guarantees that the size of a tree node is always the sum of the sizes of its two children, which simplifies our discussion.
This is also the standard balancing criteria used for weight-balanced trees~\cite{nievergelt1973binary}.

\myparagraph{Interval trees and the 1D stabbing queries}
An \emph{interval
  tree}\footnote{There exist multiple versions of interval trees. In
  this paper, we use the version described in~\cite{DCKO08}.}~\cite{DCKO08,edelsbrunner1980dynamic,mccreight1980efficient}
  organizes
  a set of $n$ intervals $S=\{s_i=(l_i,r_i)\}$ defined by their left and right
endpoints. The key on the root of the interval tree
is the median of the $2n$ endpoints.
This median divides all intervals into two categories: those completely on its
left/right, which then form the left/right subtrees recursively,
and those covering the median, which are stored in the root.
The intervals in the root are stored in two lists sorted by the left and right endpoints
respectively.
In this paper, we use red-black trees to maintain such ordered lists to support dynamic updates
and refer to them as the \emph{inner trees}. 
In the worst case, the
previous construction algorithms scan and copy $O(n)$ intervals in $O(\log n)$
levels, leading to $O(n\log n)$ reads and writes.

The interval tree can be used to answer a 1D stabbing query: given a set
of intervals, report a list of intervals covering the specific query point $p_q$.
This can be done by searching $p_q$ in the tree.
Whenever $p_q$ is smaller (larger) than the key of the current node,
all intervals in the current tree node with left (right) endpoints smaller than $p_q$ should be reported.
This can be done efficiently by scanning the list sorted by left (right) endpoints.
The overall query cost is $O(\wcost
k +\log n)$ (where $k$ is the output size).

\hide{ OLD VERSION:
Given a set
of $n$ intervals $S=\{s_i=(l_i,r_i)\}$ defined by their left and right
endpoints, a \emph{1D stabbing query} asks for the
list of intervals covering the specific query point $p_q$.

1D stabbing queries can be answered efficiently using \emph{interval
  trees}.\footnote{There exist multiple versions of interval trees. In
  this paper we use the version described in~\cite{DCKO08}.} For $n$
intervals, the interval tree chooses the median of the $2n$ endpoints
as the key on the root, which is also denoted as the \emph{splitter}.
Each interval either fully lies on one side of the splitter, or
overlaps the splitter.  For intervals that overlap the splitter, we
maintain them in two ordered lists on the node, one sorted by the left endpoint and one sorted by the right endpoint.  We then
recursively construct an interval tree for the intervals that fully on
the left (right) side of the splitter, making it the left (right)
child of the root. The base case is reached when no intervals lie
completely on one side of the splitter.  Note that each interval is
stored in exactly one tree node (with two copies in that node). To
support dynamic updates, the ordered lists of overlapping intervals
can be maintained using balanced binary search trees (BSTs), and we
refer to these BSTs as the \emph{inner trees}.  In the worst case, the
construction scans and copies $O(n)$ intervals in $O(\log n)$
level, leading to $O(n\log n)$ reads and writes.

For a stabbing query on $p_q$, we first compare $p_q$ with the root's
key. If $p_q$ is smaller (larger) than the root's key, then no
intervals in the right (left) subtree can cover $p_q$.  For the
intervals stored in the root node, we report all intervals that have
their left endpoints to the left of $p_q$.  This can be efficiently
supported by scanning the ordered list sorted by the left (right)
endpoints (or searching the corresponding inner tree), and the cost is
proportional to the output size.  We then continue the query in the
left (right) subtree recursively.  The overall query cost is $O(\wcost
k +\log n)$ (where $k$ is the output size), since we read $O(\log n)$
tree nodes and write out $k$ intervals.
}

\myparagraph{2D Range trees and the 2D range queries}
The \emph{2D range tree}~\cite{bentley1979decomposable} organizes a set of $n$ points $p=\{p_i =(x_i,y_i)\}$ on the
2D plane. It is a tree structure augmented with an inner tree, or equivalently, a two-level tree structure. The outer tree stores every point sorted by their $x$-coordinate.
Each node in the outer tree is augmented with an inner tree structure, which contains all the points in its subtree, sorted by their y-coordinate.

The 2D range tree can be used to answer the 2D range query:
given $n$ points in the 2D plane, report the list of points with $x$-coordinate
between $x_L$ and $x_R$, and $y$-coordinate between $y_B$ and $y_T$. Such range queries using range trees can be done by
two nested searches on $(x_L,x_R)$ in the outer tree and $(y_B,y_T)$ in at most $O(\log n)$ associated inner trees.
Using balanced BSTs for both the inner and outer trees, a range tree can be constructed with $O(n\log n)$ reads and writes, and each query takes $O(\log^2 n+k)$ reads and $O(k)$ writes (where $k$ is the output size). A range tree requires $O(n\log n)$ storage so the number of writes for construction is already optimal.

\hide{OLD VERSION
Given a set of $n$ points $p=\{p_i =(x_i,y_i)\}$ on the plane, the \emph{2D
range query} asks for the list of all points with $x$-coordinate
between $x_L$ and $x_R$, and $y$-coordinate between $y_B$ and $y_T$.

This query can be answered using the \emph{range tree}, which is a tree structure augmented with an inner tree, or equivalently, a two-level tree structure. The outer tree stores every point with their $x$-coordinates as the keys.
Each node in the outer tree is augmented with an inner tree structure, which contains all the points represented in its subtree, with their $y$-coordinates as the keys. 
To answer a range query, we conduct nested searches on $(x_L,x_R)$ in the outer tree and $(y_B,y_T)$ in at most $O(\log n)$ associated inner trees.
Using balanced binary search trees for both the inner and outer trees, a range tree can be constructed with $O(n\log n)$ reads and writes, and each query takes $O(\log^2 n+k)$ reads and $O(\wcost k)$ writes (where $k$ is the output size). A range tree requires $O(n\log n)$ storage so the number of writes for construction is already optimal.
}



\myparagraph{Priority search trees and 3-sided range queries}
The \emph{priority search tree}~\cite{McCreight85,DCKO08}
(priority tree for short) contains a set of $n$ points $p=\{p_i =(x_i,y_i)\}$ each with
a \emph{coordinate} ($x_i$) and a \emph{priority} ($y_i$). There are two
variants of priority trees, one is a search tree on
coordinates that keeps a heap order of the priorities as the augmented values~\cite{McCreight85,arge1999two}. 
The other one is a heap of the priorities, where each node is
augmented with a splitter between the left and right
subtrees on the coordinate dimension~\cite{McCreight85,DCKO08}.
The construction of both variants
uses $O(n\log n)$ reads and writes as shown in the original papers~\cite{McCreight85,DCKO08}.
For example, consider the second variant.
The root of a priority tree stores the point with the highest priority
in $p$.  All the other points are then evenly split into two sets by
the median of their coordinates which recursively form the
left and right subtrees. The
construction scans and copies $O(n)$ points in $O(\log n)$
levels, leading to $O(n\log n)$ reads and writes for the
construction.

Many previous results on dynamic priority search trees use the first
variant because it allows for rotation-based updates.  In this paper,
we discuss how to construct the second variant allowing
reconstruction-based updates, since it is a natural fit for our
framework.  We also show that bulk updates can be done
write-efficiently in this variant.  For the rest of this section, we
discuss the second variant of the priority tree.

The priority tree can be used to answer the 3-sided queries: given a set
of $n$ points, report all points with coordinates in the range
$[x_L,x_R]$, and priority higher than $y_B$. This can be done by
traversing the tree, skipping the subtrees whose coordinate range do
not overlap $[x_L,x_R]$, or where the priority in the root is lower
than $y_B$.  The cost of each query is $O(\wcost{}k+\log n)$ for an
output of size $k$~\cite{DCKO08}.

\hide{ OLD VERSION
Given a set of $n$ points $p=\{p_i =(x_i,y_i)\}$ on the plane
the \emph{3-sided range query} asks for the list of points with
$x$-coordinates in the range $[x_L,x_R]$, $y$-coordinates in the range
$[y_B,\infty)$. We will refer to the $y$-coordinate of a point as its
  \emph{priority}.

The \emph{priority search tree}~\cite{McCreight85,DCKO08} (priority
tree for short) can answer such queries efficiently.  There are two
variants of the priority trees, one is a search tree based on
$x$-coordinates and augments the tree node as a heap of the priorities
of the points~\cite{McCreight85,arge1999two}. 
The other variant is a heap of the priorities, and each node is
augmented with a splitter on the coordinate dimension, which
partitions the other points to the left and right
subtrees~\cite{McCreight85,DCKO08}.  The construction of both variants
uses $O(n\log n)$ reads and writes on $n$ points.


Most previous work uses the first variant because it allows for
rotation-based updates.  In this paper we discuss how to construct the
second variant of priority trees and perform updates for it, since it
is a natural fit for our framework.
In addition, we show that by using the second variant, bulk updates on
priority trees write-efficiently.
For the rest of this section, we refer to a priority tree as the
second variant described above.

The root of a priority tree stores the point with the highest priority
in $p$.  Then all the other points are evenly split into two sets by
the median of their $x$-coordinates, and they recursively form the
left and right subtrees.  The construction algorithm scans all points
to find the root, removes the corresponding point, finds the median of
the $x$-coordinates of the remaining points, splits them into two
sets, and recursively build the left and right subtrees.  The tree
contains $O(\log n)$ levels, and we read and write each point at most
once per level. In total we need $O(n\log n)$ reads and writes for the
construction.

To answer a 3-sided query, we traverse the tree, skipping the subtrees
whose $x$-coordinate range do not intersect with $[x_L,x_R]$, or whose
priority is lower than $y_B$.  The cost of each query is
$O(\wcost{}k+\log n)$ for an output of size $k$~\cite{DCKO08}.}


\hide{
\paragraph{$(\alpha,2\alpha)$-WBB Tree.} (..)
We use this technique in the pre-sorted interval tree, pre-sorted priority search tree and the range tree. It can reduce the number of writes per update at the cost of increasing the query time, hence is more applicable in the scenario where updates is significantly more frequent than queries.
Due to page limit we only explain in details for the dynamic range tree as an example.
}

\hide{
We summarize the cost of the previous work and our results in Table \ref{tab:treecost}.
\hide{
\paragraph{Augmented Tree.}
The \emph{augmented tree} structure is a search tree structure in which each node is associated with some extra information, probably reflecting the property of the subtree rooted at it. We note that the augmented value may depend on not only the entries in $u$'s subtree, but also $u$'s ancestors, siblings, etc. This may extend the definition of the augmented tree in most textbooks, but it is crucial for at least one of our applications (e.g., the interval tree and the balanced priority search tree).}
\hide{
In an augmented tree $t$, Assume we store in each node $u\in t$ an entry $e(u)\in E$, ordered by some comparison function $\prec$, an augmented tree structure will also keep track of some quantity $a(u)\in A$, called the \emph{augmented value} of $u$, which can be computed using an \emph{augmenting} function $f(t,u)$. We note that $f(t,u)$ may depend on not only the entries in $u$'s subtree, but also $u$'s ancestors, siblings, etc. This may extend the definition of the augmented tree in most textbooks, but it is crucial for at least one of our applications (e.g., the \tournament{}).
We accordingly represent an augmented tree structure using $\augt(E,\prec,A,f)$. We also use $\tree(E,\prec)$ to represent a non-augmented tree.}

\hide{\paragraph{Persistent tree using path-copying.} When implemented with path-copying, the tree structure can also be fully \emph{persistent} (or \emph{functional}), which means that any updates will not modify existing tree but create a new version, with no extra work asymptotically.}
}

%% file: presort.tex
\subsection{The Post-Sorted Construction}
\label{sec:presort}

For interval trees and priority search trees, the standard construction
algorithms~\cite{DCKO08,CLRS,edelsbrunner1980dynamic,mccreight1980efficient,McCreight85}
require $O(n\log n)$ reads and writes, even though the output is only
of linear size. This section describes algorithms for constructing
them in an optimal linear number of writes. Both algorithms first sort
the input elements by their $x$-coordinate in $O(\omega n+n\log n)$
work and $O(\log^2n)$ depth using the write-efficient comparison sort
described in Section~\ref{sec:inc-sorting}.
We now describe how to build the trees in $O(n)$ reads and writes given the sorted input.
For a range tree, since the standard tree has $O(n\log n)$ size, the classic construction algorithm is already optimal.

\myparagraph{Interval Tree}
After we sort all $2n$ coordinates of the endpoints, we can first build a perfectly-balanced binary search tree on the endpoints using $O(n)$ reads and writes and $O(\log n)$ depth.
We now consider how to construct the inner tree of each tree node.

We create a lowest common ancestor (LCA) data structure on the keys of
the tree nodes that allows for constant time queries. This can be
constructed in $O(n)$ reads/writes and $O(\log^2 n)$ depth~\cite{BV89,JaJa92}.
Each interval can then find the tree node that it belongs to using an LCA
query on its two endpoints.  We then use a radix sort on the $n$
intervals.
The key of an interval is a pair with the first value being the index
of the tree node that the interval belongs to, and the second value
being the index of the left endpoint in the pre-sorted order. The sorted result gives the interval list for each tree node sorted by left endpoints. We
do the same for the right endpoints. This step takes $O(n)$ reads/writes
overall. Finally, we can construct the inner trees from the sorted
intervals in $O(n)$ reads/writes across all tree nodes.

Parallelism is straightforward for all steps except for the radix
sort.  The number of possible keys can be $O(n^2)$, and it is not known
how to radix sort keys from such a range work-efficiently and in
polylogarithmic depth. However, we can sort a range of $O(n\log n)$ in $O(\wcost n)$ expected work and $O(\log^2 n)$ depth \whp{}~\cite{RR89}.
Hence our goal is to limit the first value into a $O(\log n)$ range.
We note that given the left endpoint of an interval, there are only $\log_2 (2n)$ possible locations for the tree node (on the tree path) of this interval. Therefore instead of using the tree node as the first value of the key, we use the level of the tree node, which is in the range $[1,\ldots, O(\log n)]$.
By radix sorting these pairs, so we have the sorted intervals (based on left or right endpoint) for each level. We observe that the intervals of each tree node are consecutive in the sorted interval list per level. This is because
for tree nodes $u_1$ and $u_2$ on the same level where $u_1$ is to the left of $u_2$, the endpoints of $u_1$'s intervals must all be to the left of $u_2$'s intervals.
Therefore, in parallel we can find the first and the last intervals of each node in the sorted list, and construct the inner tree of each node.
Since the intervals are already sorted based on the endpoints, we can build inner trees in $O(n)$ reads and writes and $O(\log^2 n)$ depth~\cite{blelloch2016just}.

\myparagraph{Priority Tree} In the original priority tree construction
algorithm, points are recursively split into sub-problems based on the
median at each node of the tree.  This requires $O(n)$ writes at each
level of the tree if we explicitly copy the nodes and pack out the
root node that is removed.  To avoid explicit copying, since the
points are already pre-sorted, our write-efficient construction
algorithm passes indices determining the range of points belonging to
a sub-problem instead of actually passing the points themselves.  To
avoid packing, we simply mark the position of the removed point in the
list as invalid, leaving a hole, and keep track of the number of valid
points in each sub-problem.

Our recursive construction algorithm works as follows. For a tree node, we
know the range of the points it represents, as well as the number of
valid points $n_v$.  We then pick the valid point with the highest
priority as the root, mark the point as invalid, find the median among
the valid points, and pass the ranges based on the median and number of
valid points (either $\lfloor(n_v-1)/2\rfloor$ or
$\lceil(n_v-1)/2\rceil$) to the left and right sub-trees, which are
recursively constructed.  The base case is when there is only one
valid point remaining, or when the number of holes is more than the
valid points.  Since each node in the tree can only cause one hole,
for every range corresponding to a node, there are at most $O(\log n)$
holes.  Since the size of the \smallmem{} is $\Omega(\log n)$, when
the number of valid points is fewer than the number of holes, we can
simply load all of the valid points into the \smallmem{} and construct
the sub-tree.

To efficiently implement this algorithm, we need to support three queries on the input list: finding the root, finding the $k$-th element in a range (e.g., the median), and deleting an element.
All queries and updates can be supported using a standard tournament tree where each interior node maintains the minimum element and the number of valid nodes within the subtree.
With a careful analysis, all queries and updates throughout the construction require linear reads/writes overall.
The details are provided in Appendix~\ref{sec:aug-app}.

The parallel depth is $O(\log^2 n)$---the bottleneck lies in removing the points. There are $O(\log n)$ levels in the priority tree and it costs $O(\log n)$ writes for removing elements from the tournament tree on each level.
For the base cases, it takes linear writes overall to load the points into the \smallmem{} and linear writes to generate all tree nodes.
The depth is $O(\log n)$.

We summarize our result in this section in Theorem \ref{thm:treebuild}.
\begin{theorem}
\label{thm:treebuild}
An interval tree or a priority search tree can be constructed with pre-sorted input in $O(\wcost{}n)$ expected work and $O(\log^2 n)$ depth \whp{} on the \anp{} model. 

\end{theorem}


%% file: dynamic.tex
\subsection{Dynamic Updates using Reconstruction-Based Rebalancing}\label{sec:aug-dyn}

Dynamic updates (insertions and deletions) are often supported on
augmented
trees~\cite{DCKO08,CLRS,edelsbrunner1980dynamic,mccreight1980efficient,McCreight85}
and the goal of this section is to support updates write-efficiently,
at the cost of performing extra reads to reduce the overall work.
Traditionally, an insertion or deletion costs $O(\log n)$ for interval
trees and priority search trees, and $O(\log^2 n)$ for range trees. In
the asymmetric setting, the work is multiplied by $\wcost{}$.
To reduce the overall work, we introduce an approach to select a
subset of tree nodes as \emph{critical} nodes, and only update the
balance information of those nodes (the
augmentations are mostly unaffected).
The selection of these critical nodes are done by the \labeling{} introduced in Section~\ref{sec:labeling}. Roughly speaking, for each tree path from the root to a leaf node, we have $O(\log_\alpha n)$ critical nodes marked such that the subtree weights of two consecutive marked nodes differ by about a factor of $\alpha \ge 2$.
By doing so, we only need to update the balancing information in the critical nodes, leading to fewer tree nodes modified in an update.


Arge et al.~\cite{arge1999two,arge2003optimal} use a similar strategy to support dynamic updates on augmented trees in the external-memory (EM) model, in which a block of data can be transferred in unit cost~\cite{AggarwalV88}.
They use a $B$-tree instead of a binary tree, which leads to a shallower tree structure and fewer memory accesses in the EM model.
However, in the \anp{} model, modifying a block of data requires work proportional to the block size, and directly using their approach cannot reduce the overall work.
Inspired by their approach, we propose a simple approach to reduce the work of updates for the \anp{} model.



\hide{
\yihan{this paragraph is mostly newly-written, trying explain the reconstruction and alpha-labeling at a high level. } The main component of our approach is \emph{reconstruction-based rebalancing} using the \emph{\labeling{}} technique.
The reconstruction-based rebalancing is used for avoiding writes in rotations. It defers rebalancing to where newly-added or -deleted nodes take a constant fraction of the current tree size, such that the number of writes amortized to each update is a constant.
In Section~\ref{sec:presort}, we have shown that given the sorted order of all (end)points, an interval or priority tree can be constructed in a linear number of reads and writes.
We can always obtain the sorted order via the tree structure, so when imbalance occurs, we can afford to reconstruct the whole subtree in reads and writes proportional to the subtree size and polylogarithmic depth. The \labeling{} aims at selecting a subset of tree nodes on an augmented tree, and only maintain balancing information as well as augmented values of these nodes, such that the overall number of writes of each update is reduced.
Roughly speaking, it selects every $\alpha \ge 2$ nodes on each tree path, such that the update cost is saved by a factor of $\alpha$. Altogether, combining reconstruction-based rebalancing with \labeling{} gives a unified approach for different augmented trees: interval trees, priority search trees, and range trees.
}

The main component of our approach is \emph{reconstruction-based rebalancing} using the \emph{\labeling{}} technique.
We can always obtain the sorted order via the tree structure, so when imbalance occurs, we can afford to reconstruct the whole subtree in reads and writes proportional to the subtree size and polylogarithmic depth.
This gives a unified approach for different augmented trees: interval trees, priority search trees, and range trees.

We introduce
the \labeling{} idea in Section~\ref{sec:labeling}, the rebalancing algorithm in Section~\ref{sec:rebalancing}, and its work analysis in Section~\ref{sec:labelingcost}. We then discuss the maintenance of augmented values for different applications in Section~\ref{sec:aug-maintain}.
We mention how to parallelize bulk updates in Section~\ref{sec:bulk}.

\subsubsection{$\alpha$-Labeling}
\label{sec:labeling}
The goal of the \labeling{} is to maintain the balancing information at only a subset of tree nodes, the critical nodes, such that the number of writes per update is reduced.
Once the augmented tree is constructed, we label the node as a \critical{} node if for some integer $i\ge0$, (1) its subtree weight is between $2\alpha^i$ and $4\alpha^i-2$ (inclusive); or (2) its subtree weight is $2\alpha^i-1$ and its sibling's subtree weight is $2\alpha^i$.
All other nodes are \emph{\secondary{}} nodes.
As a special case, we always treat the root as a virtual \critical{} node, but it does not necessary
satisfy the invariants of \critical{} nodes.
Note that all leaf nodes are \critical{} nodes in \labeling{} since they always have subtrees of weight 2. 
When we label a \critical{} node, we refer to its current subtree weight (which may change after insertions/deletions) as its \emph{initial weight}.
Note that after the augmented tree is constructed, we can find and mark the critical nodes in $O(n)$ reads/writes and $O(\log n)$ depth.
After that, we only maintain the subtree weights for these \critical{} nodes, and use their weights to balance the tree.

\begin{fact}
\label{fact:initw}
For a critical node $A$, $2\alpha^i-1\le |A|\le 4\alpha^i-2$ holds for some integer $i$.
\end{fact}
This fact directly follows the definition of the critical node.

For two \critical{} nodes $A$ and $B$, if $A$ is $B$'s ancestor and there is no other \critical{} node on the tree path between them, we refer to $B$ as $A$'s \emph{\critical{} child}, and $A$ as $B$'s \emph{\critical{} parent}.
We define a \emph{\critical{} sibling} accordingly.

We show the following lemma on the initial weights.
\begin{lemma}\label{lemma:critical-static}
For any two \critical{} nodes $A$ and $B$ where $A$ is $B$'s \critical{} parent, their initial weights satisfy $\max\{(\alpha/2)|B|,2|B|-1\}\le |A|\le  (2\alpha+1) |B|$.
\end{lemma}
\begin{proof}
Based on Fact \ref{fact:initw}, we assume $2\alpha^i-1\le |A|\le 4\alpha^i-2$ and $2\alpha^j-1\le |B|\le 4\alpha^j-2$ for some integers $i$ and $j$.
We first show that $i=j+1$.
It is easy to check that $j$ cannot be larger than or equal to $i$.
Assume by contradiction that $j<i-1$.
With this assumption, we will show that there exists an ancestor of $B$, which we refer to it as $y$, which is a critical node.
The existence of $y$ contradicts the fact that $A$ is $B$'s \critical{} parent.
We will use the property that for any tree node $x$ the weight of its parent $p(x)$ is $2|x|-1\le |p(x)|\le 2|x|+1$.

Assume that $B$ does not have such an ancestor $y$. Let $z$ be the ancestor of $B$ with weight closest to but no more than $2\alpha^{i-1}$. We consider two cases: (a) $|z|\leq 2\alpha^{i-1}-2$ and (b) $|z|=2\alpha^{i-1}-1$. In case (a)  $z$'s parent $p(z)$ has weight at most $2|z|+1 = 4\alpha^{i-1}-3$. $|p(z)|$ cannot be less than  $2\alpha^{i-1}$ by definition of $z$, and so $y=p(z)$, leading to a contradiction. In case (b), $z$'s sibling does not have weight $2\alpha^{i-1}$, otherwise $y=z$. However, then $|p(z)| \leq 2|z| = 4\alpha^{i-1}-2$, and either $z$ is not the ancestor with weight closest to $2\alpha^{i-1}$ or $y=p(z)$.


Given $i=j+1$, we have $(\alpha/2)|B|\le |A|\le (2\alpha+1) |B|$ (by plugging in $2\alpha^i-1\le |A|\le 4\alpha^i-2$ and $2\alpha^{i-1}-1\le |B|\le 4\alpha^{i-1}-2$).
Furthermore, since $A$ is $B$'s ancestor, we have $2|B|-1\le |A|$.
Combining the results proves the lemma.
\end{proof}

\subsubsection{Rebalancing Algorithm based on $\alpha$-Labeling}\label{sec:aug-rebalance}
\label{sec:rebalancing}
We now consider insertions and deletions on an augmented tree.
Maintaining the augmented values on the tree are independent of our \labeling{} technique, and differs slightly for each of the three tree structures. We will further discuss how to maintain augmented values in Section~\ref{sec:aug-maintain}.


We note that deletions can be handled by marking the deleted objects without actually applying the deletion, and reconstructing the whole subtree once a constant fraction of the objects is deleted. Therefore in this section, we first focus on the insertions only.
We analyze single insertions here, and discuss bulk insertions later in Section~\ref{sec:bulk}.
Once the subtree weight of a \critical{} node $A$ reaches twice
the initial weight $s$, we reconstruct the whole subtree, label the \critical{} nodes within the subtree, and recalculate the initial weights of the new critical nodes.
An exception here is that, if $s\le 4\alpha^i-2$ and $2\alpha^{i+1}-1\le 2s$ for a certain $i$, we do not mark the new root since otherwise it violates the bound stated in Lemma~\ref{lemma:critical-dynamic} (see more details in Section \ref{sec:labelingcost}) with $A$'s \critical{} parent. 
After this reconstruction, $A$'s original \critical{} parent gets one extra \critical{} child, and the two affected children now have initial weights the same as $A$'s initial weight.
If imbalance occurs at multiple levels, we reconstruct the topmost tree node.
An illustration of this process is shown in Figure~\ref{fig:aug}.

\input{aug-figure}

We can directly apply the algorithms in Section~\ref{sec:presort} to reconstruct a subtree as long as we have the sorted order of the (end)points in this subtree.
For interval and range trees, we can acquire the sorted order by traversing the subtree. 
using linear work and $O(\log n)$ depth~\cite{SGBFG2015,BBFGGMS16}. 
For priority trees, since the tree nodes are not stored in-order, we need to insert all interior nodes into the tree in a bottom-up order based on their coordinates (without applying rebalancing) to get the total order on coordinates of all points (the details and cost analysis can be found in Appendix~\ref{sec:aug-app}).
After we have the sorted order, a subtree of weight $n$ can be constructed in $O(\wcost n)$ work and $O(\log^2 n)$ depth.

As mentioned, we always treat the root as a virtual \critical{} node, but it does not necessary
satisfy the invariants of \critical{} nodes.
By doing so, once the weight of the whole tree doubles, we reconstruct the entire tree.
We need $\Omega(n)$ insertions for one reconstruction on the root (there can be deletions).
The cost for reconstruction is $O(\wcost{}n)$ for interval trees and priority trees, and $O(\wcost n\log_\alpha n)$ for range trees (shown in Section~\ref{sec:aug-maintain}).
The amortized cost is of a lower order compared to the update cost shown in Theorem~\ref{thm:aug-dynamic}. 

\subsubsection{Cost Analysis of the Rebalancing}
\label{sec:labelingcost}
\hide{
\begin{invariant}
For a \critical{} tree node $A$ with subtree weight $S_A$, the subtree weight $S_B$ of $A$'s \critical{} sibling $B$ satisfies $S_A<4S_B$.
\end{invariant}
}

To show the rebalancing cost, we first prove some properties about our
dynamic augmented trees.

\begin{lemma}\label{lemma:critical-dynamic}
In a dynamic augmented tree with \labeling{},
we have $\max\{(\alpha/4)|B|,(3/2)|B|-1\}\le |A|\le  (4\alpha+2) |B|$
for any two \critical{} nodes $A$ and $B$ where $A$ is $B$'s \critical{} parent.

\end{lemma}
\begin{proof}
For any \critical{} node $A$ in the tree, the subtree weight of its \critical{} child $B$ can grow up to a factor of 2 of $B$'s initial weight, after which the subtree is reconstructed to two new critical nodes with the same initial weight of $B$. $A$'s weight can grow up to a factor of 2 of $A$'s initial weight, without affecting $B$'s weight (i.e., all insertions occur in $A$'s other \critical{} children besides $B$).
Combining these observations with the result in Lemma~\ref{lemma:critical-static} shows this lemma except for the $(3/2)|B|-1\le |A|$ part.
Originally we have $2|B|-1\le |A|$ after the previous reconstruction.
$|A|$ grows together when $|B|$ grows, and right before the reconstruction of $B$ we have $(3/2)|B|-1\le |A|$.
\end{proof}

Lemma~\ref{lemma:critical-dynamic} shows that each critical node has at most $4\alpha+2$ \critical{} children, and so that there are at most $4n+1$ \secondary{} nodes to connect them.
This leads to the following corollary.
\begin{corollary}\label{cor:children}
The length of the path from a \critical{} node to its \critical{} parent is at most $4\alpha+1$.
\end{corollary}

 Combining Lemma~\ref{lemma:critical-dynamic} and Corollary~\ref{cor:children} gives the following result.

\begin{corollary}\label{cor:treedepth}
For a leaf node in a tree with \labeling{}, the tree path to the root contains $O(\log_\alpha n)$ \critical{} nodes and $O(\alpha\log_\alpha n)$ nodes.
\end{corollary}

Corollary~\ref{cor:treedepth} shows the number of reads during locating a node in an augmented tree, and the number of \critical{} nodes on that path.

With these results, we now analyze the cost of rebalancing for each insertion.
For a \critical{} node with initial weight $W$, we need to insert at least another $W$ new nodes into this subtree before the next reconstruction of this \critical{} node.
Theorem~\ref{thm:treebuild} shows that the amortized cost for each insertion in this subtree is therefore $O(\wcost{})$ on this node.
Based on Corollary~\ref{cor:treedepth}, the amortized cost for each insertion contains $O(\log_\alpha n)$ writes and $O(\alpha\log_\alpha n)$ reads.
In total, the work per insertion is $O((\alpha+\wcost)\log_\alpha n)$, since we need to traverse $O(\alpha\log_\alpha n)$ tree nodes, update $O(\log_\alpha n)$ subtree weights, and amortize $O(\wcost\log_\alpha n)$ work for reconstructions.

We note that any interleaving insertions can only reduce the amortized cost for deletions. Therefore, both the algorithm and the bound can be extended to any interleaving sequence of insertions and deletions.
Altogether, we have the following result, which may be of independent interest.

\begin{theorem}
Using {reconstruction-based rebalancing} based on the {\labeling{}} technique, the amortized cost of each update (insertion or deletion) to maintain the balancing information on a tree of size $n$ is $O((\wcost{}+\alpha)\log_{\alpha} n)$.
\end{theorem}

\subsubsection{Handling Augmented Values}\label{sec:aug-maintain}
Since the underlying tree structure is still binary, minor changes to the trees are required for different augmentations. 

\paragraph{Interval trees.} We do not need any changes for the interval tree.
Since we never apply rotations, we directly insert/delete the interval in the associated inner tree with a cost of $O(\log n+\wcost{})$.

\paragraph{Range trees.} For the range tree, we only keep the inner trees for the \critical{} nodes.
As such, the overall augmentation weight (i.e., overall weights of all inner trees) is $O(n\log_\alpha n)$. 
For each update, we insert/delete this element in $O(\log_\alpha n)$ inner trees (Corollary~\ref{cor:treedepth}), and the overall cost is $O((\log n+\wcost{})\log_\alpha n)$.
Then each query may look into no more than $O(\alpha\log_\alpha n)$ inner trees each requiring $O(\log n)$ work for a 1D range query.
The overall cost for a query is therefore $O(\wcost{}k+\alpha\log_\alpha n\log n)$.

\paragraph{Priority trees.}
For insertions on priority trees, we search its coordinate in the tree and put it where the current tree node is of lower priority than the new point.
The old subtree root is then recursively inserted to one of its subtrees.
The cost can be as expensive as $O(\wcost \alpha \log_\alpha n)$ when a point with higher priority than all tree nodes is inserted.
To address this, points are only stored in the \critical{} nodes, and the \secondary{} nodes only partition the range, without holding points as augmented values.
This can be done by slightly modifying the construction algorithm in Section~\ref{sec:presort}.
During the construction, once the current node is a \secondary{} node, we only partition the range, but do not find the node with the highest priority.
Since all leaf nodes are \critical{}, the tree size is affected by at most a factor of $2$.
With this approach, each insertion modifies at most $O(\log_\alpha n)$ nodes, and so the extra work per insertion for maintaining augmented data is $O((\alpha+\wcost)\log_\alpha n)$. 
A deletion on priority trees can be implemented symmetrically, and can lead to cascading promotions of the points.
Once the promotions occur, we leave a dummy node in the original place of the last promoted point, so that all of the subtree sizes remain unchanged (and the tree is reconstructed once half one the nodes are dummy).
The cost of a deletion is also $O((\alpha+\wcost)\log_\alpha n)$.

Combining the results above gives the following theorem.
\begin{theorem}\label{thm:aug-dynamic}
Given any integer $\alpha\ge 2$, an update on an interval or priority search tree requires $O((\wcost +\alpha)\log_\alpha n)$ amortized work and a query costs $O(\wcost k+\alpha\log_\alpha n)$; for a 2D range tree, the query and amortized update cost is $O(\left( \alpha \log n+\wcost \right)\log_{\alpha}n)$ and $O(\wcost k + \alpha \log_\alpha n\log n)$.
\end{theorem}


\subsubsection{Bulk Updates}
\label{sec:bulk}

One of the benefits of our reconstruction-based approach is that, bulk updates on these augmented trees can be supported directly.
In the case we change the inner trees of interval and range trees as treaps.
For a treap of size $n$ and a bulk update of size $m$, the expected cost of inserting or deleting this bulk is $O(\wcost{}m+m\log(n/m))$ using treaps~\cite{Gu2018}, and the depth is $O(\log m\log n)$ \whp{}~\cite{blelloch2016just,Sun2016PAM,sun2018parallel}.
We note that there are data structures supporting bulk updates in logarithmic expected depth~\cite{blelloch1999pipelining,AS16}, but we are not sure how to make them write-efficient.

Again deletions are trivial.
For interval and range trees, we can just mark all the objects in parallel but apply deletions to inner trees, which requires constant writes per deletion.
For priority tree, we delete can delete points in a top-down manner, using $O((\alpha+\wcost{})\log_\alpha n)$ work per point and $O(\alpha\log_\alpha n\log n)$ depth.
We now sketch an outline on the bulk insertion.

Assume the bulk size is $m$ and less than $n$ since otherwise we can afford to reconstruct the whole augmented tree.
We first sort the bulk using $O(m\log m+\wcost{}{m})$ work and $O(\log^2 n)$ depth.
Then we merge the sorted list into the augmented tree recursively, and check each \critical{} node in the tree in a top-down manner.

At any time and for a \critical{} node, we use binary search to decide the new objects inserted in this subtree.
If the overall size of the subtree and the newly added objects overflows $4\alpha^i-2$, we reconstruct the subtree by first flattening the tree nodes, merging with new nodes, rebuilding the subtree using the algorithm in Section~\ref{sec:presort}, and marking the \critical{} nodes in this subtree.
To guarantee Lemma~\ref{lemma:critical-dynamic}, we do not mark the \critical{} nodes with subtree size greater than or equal to $2\alpha^{i+1}-1$.
By doing so, the proof of Lemma~\ref{lemma:critical-dynamic} still holds.
The whole process takes $O(n')$ operations in $O(\log n')$ depth for a subtree with $n'$ nodes.
Note that at least $O(n')$ nodes are inserted between two consecutive reconstructions so that the cost can be amortized.
Otherwise, we just recursively check all the \critical{} children of this node.
Once there are no objects within one subtree, we stop the recursion in this subtree.

Note that the range of a binary search can be limited by the range of the \critical{} parent.
The overall cost for all binary searches is $O(\alpha m \log (n/m))$~\cite{blelloch2016just} (no writes), and the depth is $O(\alpha\log m\log_\alpha n)$.
In summary, the amortize work for merging $m$ new objects in to the tree is $O(\alpha m \log (n/m)+\wcost{}m\log_\alpha n)$, and the depth is $O(\alpha\log m\log_\alpha n)$.

We now discuss the bulk updates for the augmented values.
Again for interval trees and range trees, we can just merge all inserted objects into the corresponding inner trees.
Using treaps as the inner trees, merging $m'$ objects to a search tree with $n'$ takes $O(m'\log (n'/m')+\wcost{}m')$ work and $O(\wcost{}\log m'\log n')$ depth.
As a result, for interval and range trees, the work per object in the bulk updates is always no more than the single insertion, and the depth is polylogarithmic for any bulk size.

The bulk update for priority trees is similar to the constructions.
Once there exists an inserted point with higher priority than the root node, we replace root node with this inserted node, and insert the original root node into the corresponding subtree.
Then we leave a hole in the inserted list and recursively apply this process.
This process terminates at the time either the subtree root overflows, we reach a leaf node, or there are more holes than new objects.
The maintenance of augmentations of priority tree takes $O((\alpha+\wcost{})m\log_\alpha n)$ amortized work and $O(\alpha\log_\alpha^2 n)$ depth.

%% file: aug-figure.tex
\begin{figure}
\begin{center}
  \includegraphics[width=.8\columnwidth]{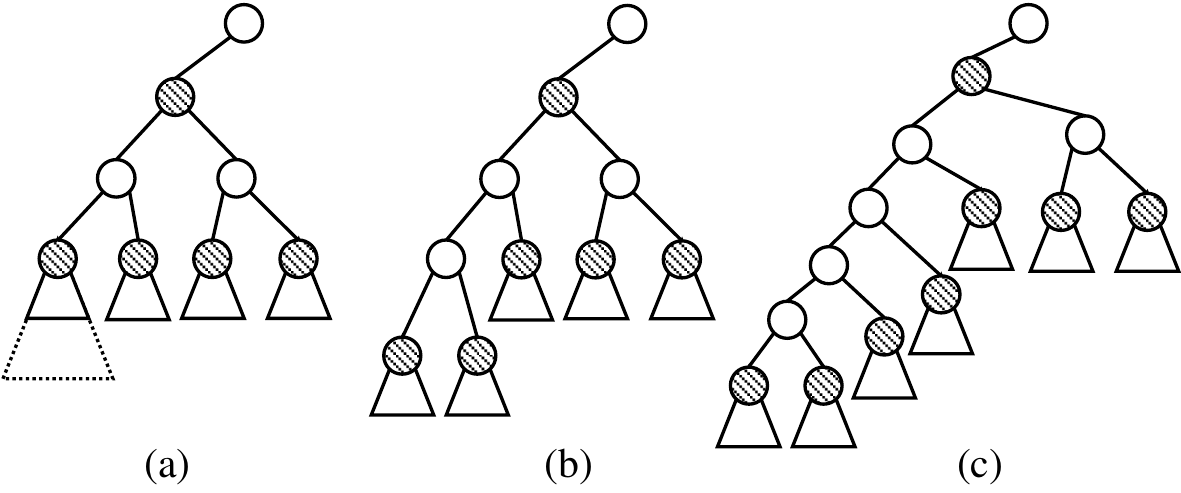}
\end{center}\vspace{-1.2em}
\caption[An illustration of rebalancing based on \labeling{}.]
{An illustration of rebalancing based on \labeling{}. The \critical{} nodes are shaded.
The case after construction is shown in~(a) with solid borders.
After some insertions, the size of one of the subtrees grows to twice its initial weight (dashed lines in (a)), so the algorithm reconstructs the subtree, as shown in~(b).
As we keep inserting new nodes along the left spine, the tree will look like what is shown in~(c), but Lemma~\ref{lemma:critical-dynamic} guarantees that the subtree of the topmost \critical{} node will be reconstructed before it gets more than $4\alpha+2$ \critical{} children.
The lemma also guarantees that on the path from a critical node to any of its critical children, there can be at most $4\alpha-1$ \secondary{} nodes.
}
\label{fig:aug}
\end{figure}

%% file: aug-appendix.tex
\section{Additional Details for Augmented Trees}\label{sec:aug-app}

In this section we provide omitted details for the write-efficient augmented trees introduced in Section~\ref{sec:augtree}.

\myparagraph{The tournament tree for constructing priority trees}
In Section~\ref{sec:presort}, we mentioned that a tournament tree on a list is used to support the RangeMin and the $k$-th element in a range, and at the meantime an element can be removed.
Each interior tree node maintains the minimum element and the number of valid nodes within the subtree.
We now show that given the tree of size $n$, answering all queries in construction uses linear reads and writes.

For a RangeMin query on $(x,y)$, we start from the left corresponding to $x$, keep going up on the tree until the node that its subtree contains $y$, and traverse the tree to find $y$.
In this process, we use the maintained values in the tree nodes to update the RangeMin when the corresponding ranges of the subtrees are within $(x,y)$.
The reads of such a query is $O(\log (y-x+1))$.
We can query the $k$-th element in a range similarly.

Since the tree is fully balanced, the tree height is $\log_2 n$.
In the $i$-th level (the root is the first level), there are $O(2^i)$ queries and the sum of the query ranges is $O(n)$.
The overall query cost on one level is maximized when all the query ranges are the same, which is $O(2^i\cdot \log(n/2^i))$.
The overall cost across all levels is $\sum_{i=1}^{\log_2 n}O(2^i\cdot \log(n/2^i))=O(n)$.

Deleting an element na\"ively in a tournament tree costs $O(\log n)$ writes in the worst case.
Our observation is that, once we delete an element of a tree node corresponding to a range $(x,y)$, we know that all further queries are either entirely within $(x,y)$ or disjoint $(x,y)$.
We therefore only update the ancestors of the deleted nodes whose range is within $(x,y)$, and there are at most $O(\log(y-x+1))$ of such ancestors.
The overall writes required in all deletions has the same form as the overall reads in the queries, which is $O(n)$.

\myparagraph{Ordering the nodes within the subtree of a priority tree}
Since our priority tree is not a search tree by default, we need an extra step to obtain the ordering of the points in a subtree.
This can be trivially achieved by inserting the points in the \critical{} nodes into their subtrees in a bottom-up manner (without balancing the tree).
Inserting an object into a tree requires $O(1)$ writes, so the overall writes are linear.
By Corollary~\ref{cor:treedepth}, the subtree depth is $O(\alpha\log_\alpha m)$ for a subtree of size $m$.
After the all insertions, the tree depth can be increased by at most $O(\log_\alpha m)$.
For an critical node $A$ such that $2\alpha^i-1\le |A|\le 4\alpha^i-2$ for some integer $i$, the number of reads required to find the leaf node is proportional to the tree depth, $O(\alpha i)$.
By Lemma~\ref{lemma:critical-dynamic}, the number \critical{} nodes with the same $i$ decreases geometrically with the increasing of $i$, the overall reads is asymptotically bounded by the level where $i=1$.
The overall number of reads is therefore $O(\alpha m)$.
The total cost to get the ordering is $O((\wcost+\alpha)m)$ for a subtree of size $m$.

\myparagraph{Range tree construction based on \labeling{}}
When building a range tree based on \labeling{}, the we skip the construction of the inner trees for the \secondary{} nodes.
Note that whether a tree node is \critical{} can be checked once it is created.
For a \critical{} node, we can get the inner tree nodes by applying an ordered filter of its \critical{} parent's inner tree.
By Lemma~\ref{lemma:critical-static} and the result in~\cite{BBFGGMS16}, this step costs $O((\alpha+\wcost{})s)$ where $s$ is the inner tree size.
As shown in Section~\ref{sec:aug-dyn}, the overall inner tree size is $O(n\log_\alpha n)$, so the cost to generate them is $O((\alpha+\wcost{})\, n\log_\alpha n)$.

\myparagraph{Other queries on our augmented trees}
In this paper we mainly focused on the queries on reporting a full list of all queried elements.
Indeed, many other similar queries can be handled with a variant of our structures. 
For example, counting or weighted sum queries on interval trees and range trees can be answered by augmenting the inner trees with the count or weighted sum of all elements in the subtree, possibly with the \labeling{}.